\documentclass[format=acmsmall,screen=true]{acmartvar}

\usepackage{acm-ec-19}

\usepackage{booktabs}

\usepackage[ruled,vlined,linesnumbered]{algorithm2e}
\providecommand{\DontPrintSemicolon}{\dontprintsemicolon}
\SetKwFor{WP}{}{:}{}

\usepackage{amsmath,amsfonts,amsthm,amssymb}
\usepackage[retainorgcmds]{IEEEtrantools}
\usepackage{xcolor}

\usepackage{url}
\usepackage{tikz}

\usepackage{enumitem}

\usepackage[normalem]{ulem}
\usepackage{changepage}
\usepackage{microtype}

\usepackage{graphicx} 
\makeatletter
\newcommand*{\algotitle}[2]{%
  \stepcounter{algocf}%
  \hypertarget{algocf.title.\theHalgocf}{}%
  \NR@gettitle{#1}%
  \label{#2}%
  \addtocounter{algocf}{-1}%
}
\makeatother

\DeclareMathOperator{\alg}{\textnormal{\textsc{alg}}}
\DeclareMathOperator{\opt}{\textnormal{\textsc{opt}}}

\DeclareMathOperator*{\argmax}{arg\,max}

\newcommand{\mysetminusD}{\hbox{\tikz{\draw[line width=0.6pt,line cap=round] (3pt,0) -- (0,6pt);}}}
\newcommand{\mysetminusT}{\mysetminusD}
\newcommand{\mysetminusS}{\hbox{\tikz{\draw[line width=0.45pt,line cap=round] (2pt,0) -- (0,4pt);}}}
\newcommand{\mysetminusSS}{\hbox{\tikz{\draw[line width=0.4pt,line cap=round] (1.5pt,0) -- (0,3pt);}}}
\newcommand{\mysetminus}{\mathbin{\mathchoice{\mysetminusD}{\mysetminusT}{\mysetminusS}{\mysetminusSS}}}
\allowdisplaybreaks

\newtheorem{theorem}{Theorem}[section]

\newtheorem{claim}[theorem]{Claim}
\theoremstyle{acmdefinition}
\newtheorem{remark}[theorem]{Remark}

\usepackage{ifthen}
% ...
\newenvironment{rtheorem}[3][]{%
	\noindent\ifthenelse{\equal{#1}{}}{\sc #2 #3.}{\sc #2 #3 (#1)}%
	\begin{it}}{\end{it}}

\begin{document}
		

\title[Budget-Feasible Mechanism Design for Non-Monotone Submodular Objectives]{Budget-Feasible Mechanism Design for Non-Monotone Submodular Objectives: Offline and Online}  

	\author{Georgios Amanatidis}
\authornote{G.~Amanatidis and P.~Kleer are supported by the NWO Gravitation Project NETWORKS, Grant Number 024.002.003.}
\affiliation{%
	\institution{Centrum Wiskunde \& Informatica (CWI)}
	\country{The Netherlands}
}
\email{Georgios.Amanatidis@cwi.nl}

\author{Pieter Kleer}
\authornotemark[1]
\affiliation{%
	\institution{Centrum Wiskunde \& Informatica (CWI)}
	\country{The Netherlands}
}
\email{P.S.Kleer@cwi.nl}

\author{Guido Sch{\"a}fer}
\affiliation{%
	\institution{Centrum Wiskunde \& Informatica (CWI) and Vrije Universiteit Amsterdam}
	\country{The Netherlands}
}
\email{G.Schaefer@cwi.nl}

\begin{abstract}
\bigskip\bigskip

The framework of budget-feasible mechanism design studies procurement auctions where the auctioneer (buyer) aims to maximize his valuation function subject to a hard budget constraint. We study the problem of designing truthful mechanisms that have good approximation guarantees and never pay the participating agents (sellers) more than the budget. We focus on the case of general (non-monotone) submodular valuation functions and derive the first truthful, budget-feasible and $O(1)$-approximation mechanisms that run in polynomial time in the value query model, for both offline and online auctions. Since the introduction of the problem by Singer \citep{Singer10}, obtaining efficient mechanisms for objectives that go beyond the class of monotone submodular functions has been elusive. Prior to our work, the only $O(1)$-approximation mechanism known for non-monotone submodular objectives required an exponential number of value queries. 

At the heart of our approach lies a novel greedy algorithm for non-monotone submodular maximization under a knapsack constraint. Our algorithm builds two candidate solutions simultaneously (to achieve a good approximation), yet ensures that agents cannot jump from one solution to the other (to implicitly enforce truthfulness). Ours is the first mechanism for the problem where---crucially---the agents are not ordered according to their marginal value per cost. This allows us to appropriately adapt these ideas to the online setting as well. 

To further illustrate the applicability of our approach, we also consider the case where additional feasibility constraints are present, e.g., at most $k$ agents can be selected. We obtain $O(p)$-approximation mechanisms for both monotone and non-monotone submodular objectives, when the feasible solutions are independent sets of a $p$-system. With the exception of additive valuation functions, no mechanisms were known for this setting prior to our work. 
Finally, we provide lower bounds suggesting that, when one cares about non-trivial approximation guarantees in polynomial time, our results are  asymptotically best possible. 
\end{abstract}
\maketitle

\newpage

\section{Introduction}
\label{sec:intro}
We consider the problem of designing \emph{budget-feasible mechanisms} for a natural model of procurement auctions. In this model, an auctioneer is interested in buying services (or goods) from a set of agents $A$. Each agent $i \in A$ specifies a cost $c_i$ to be paid by the buyer for using his service; crucially, these costs are assumed to be private information. The auctioneer has a budget $B$ and a valuation function $v(\cdot)$, where $v(S)$ specifies the value derived from the services of the agents in $S \subseteq A$. Given the (reported) costs of the agents, the goal of the auctioneer is to choose a \emph{budget-feasible} subset $S \subseteq A$ of the agents, such that the valuation $v(S)$ is maximized. 
Budget-feasibility here means that $\sum_{i \in S} p_i \le B$, where $p_i$ is the payment issued from the mechanism to agent $i$.

Note that the agents might try to extract larger payments from the mechanism by misreporting their actual costs---which of course is undesirable from the auctioneer's perspective. 
The goal, therefore, is to design budget-feasible mechanisms that \emph{(i)} elicit truthful reporting of the costs by all agents, and \emph{(ii)} achieve a good approximation with respect to the optimal value for the auctioneer. What makes the problem so intriguing is the fact that truthfulness and budget-feasibility are two directly conflicting goals, since the former is achieved by paying as much as needed to make agents indifferent to lying (see Lemma \ref{lem:myerson}).
Indicatively, the use of the celebrated truthful VCG mechanism in this setting completely fails with respect to keeping the payments bounded \citep{Singer10}.

The problem of designing budget-feasible mechanisms was introduced by Singer \citep{Singer10} and has received a lot of attention, both because of its theoretical appeal and of its relevance to several emerging application domains. A prominent such application is in crowdsourcing marketplaces (such as Mechanical Turk, Figure Eight and Clickworker) which provide online platforms to procure workforce (see \citep{AnariGN14,GoelNS14,JT17}). Another application is in the context of influence maximization in social networks, where one seeks to select influential users (see \citep{Singer12,ABM16}). 

We focus on the design of budget-feasible mechanisms for the general class of \emph{non-monotone submodular} valuation functions. Submodular objectives constitute an important class of valuation functions as they satisfy the property of diminishing returns, which naturally arises in many settings. 
Most existing works make the assumption that the valuation functions are \emph{monotone} (non-decreasing), i.e., $v(S)\leq v(T)$ for $S\subseteq T$.
Although the monotonicity assumption makes sense in certain applications, there are several examples where it is violated. 
For example, in the context of influence maximization in social networks, adding more users to the selected set may sometimes result in negative influence (see \citep{BFO10}). 
The most prominent example of a non-monotone submodular objective studied in our setting is the \emph{budgeted max-cut problem}~\citep{DobzinskiPS11,ABM17}, where  $v(\cdot)$ is determined by the cuts of a given graph. 

A natural generalization of this framework is to assume that the space of feasible sets has some structure, e.g., the feasible sets form a matroid. This variant has been studied only for additive valuation functions \citep{ABM16,LeonardiMSZ17}, despite its wide range of applications varying from team formation to spectrum markets (see~\citep{LeonardiMSZ17}). Here we study the problem for monotone and non-monotone submodular objectives under $p$-system constraints.

The purely algorithmic versions of these mechanism design problems ask for the maximization of a (non-monotone) submodular function subject to the constraint that the total cost of the selected agents does not exceed the budget; often referred to as a \emph{knapsack constraint}.
These problems are typically NP-hard, hence our focus  is on approximation algorithms that compute a close to optimal solution in polynomial time. From an algorithmic point of view, most of these problems are well-understood and admit good approximations. However, it is not clear how to appropriately convert these algorithms into truthful, budget-feasible mechanisms and, up to this work, this goal had been elusive. 
Our results illustrate that for the mechanism design problems it is possible to achieve 
the same asymptotic guarantees that are known for their algorithmic counterparts.
\medskip

\noindent {\bf Our Contributions.} 
We derive the first budget-feasible and $O(1)$-appro\-xi\-mate mechanisms for non-monotone submodular objectives, both for the offline and the online setting.
Our results for the online setting hold for the well-studied \emph{secretary  model}, where the agents arrive in a uniformly random order. 
Our mechanisms run in polynomial time in the value query model.
The highlights of this work are as follows:

\begin{itemize}[leftmargin=0.2cm,itemindent=.3cm,labelwidth=\itemindent,labelsep=0cm,align=left,itemsep=6pt,topsep=4pt]
\item We obtain the first universally truthful, budget-feasible $O(1)$-approximation mechanism for non-monotone submodular objectives in the value query model. 
	
\item We derive the first universally truthful, budget-feasible $O(1)$-approxi\-ma\-tion \emph{online} mechanism for non-monotone submodular objectives. As a consequence, we obtain the first $O(1)$-approxi\-ma\-tion algorithm for the non-monotone \emph{Submodular Knapsack Secretary Problem}  (see also Remark \ref{rem:SKS}), a budget constrained variant of the infamous Secretary Problem.
	
\item We give universally truthful, budget-feasible $O(p)$-approximation mechanisms for both monotone and non-monotone submodular objectives, when the feasible solutions are independent sets of a $p$-system. 
Beyond the additive case, nothing was known for this constrained setting.
	
\item We provide lower bounds illustrating that asymptotically our results are \emph{best possible}. On a high level, only trivial guarantees can be achieved in polynomial time if one goes beyond the class of general submodular functions or imposes constraints beyond downward closed systems. 
\end{itemize}

\medskip

\noindent
\textbf{Technical Challenges.}
It should be noted that for monotone submodular objectives all known mechanisms essentially use the same greedy subroutine introduced by Singer \citep{Singer10}: Sort all agents in decreasing order of marginal value per cost and pick as many agents as possible before hitting some carefully selected threshold. This is a simplified version of the optimal greedy algorithm of Sviridenko \shortcite{Sviridenko04} and indeed gives non-trivial approximation guarantees. Further, due to its simplicity it also has the other desired properties of truthfulness, individual rationality, and budget-feasibility. 
While this whole framework might feel somewhat straightforward, the existing literature on budget-feasible mechanisms suggests that there is a frail balance between simplicity and performance here. Only ``naive'' algorithmic ideas, like greedy, seem to have any hope generating truthful mechanisms that are robust subject to cost changes and, thus, budget-feasible. 

Unfortunately, it is easy to construct examples where running such a greedy algorithm for a non-monotone objective  results in a solution of arbitrarily poor quality. 
The algorithmic state-of-the-art for non-monotone submodular maximization under a knapsack constraint, e.g., \citep{FeldmanNS11,KulikST13,ChekuriVZ14}, provides us with quite involved algorithms on continuous relaxations of the problem that seem very unlikely to yield monotone allocation rules, and thus truthful mechanisms. The only simple (and deterministic) exception is the two-pass greedy algorithm of Gupta et al.~\shortcite{GuptaRST10}, where it is shown that running Sviridenko's greedy algorithm twice and then maximizing without the knapsack constraint is sufficient to get a deterministic $6$-approximation algorithm.\footnote{In fact, that algorithm has an approximation ratio of $4+\alpha$, where $\alpha$ is the approximation ratio of any deterministic algorithm for the unconstrained maximization of non-monotone submodular functions. Recently, Buchbinder and Feldman \shortcite{BuchbinderF18} suggested a deterministic 2-approximation algorithm for the unconstrained problem, hence the ratio of 6.}  
Despite being significantly simpler, however, this two-pass greedy algorithm still suffers with respect to monotonicity. 

More recently, several simple randomized greedy approaches for maximizing non-monotone submodular objectives subject to other (i.e., non-knapsack) constraints were proposed \citep{FeldmanZ18,CGQ15,BuchbinderFNS14,FeldmanHK17}. However, these approaches are also not applicable here. In its simplest version such a random greedy algorithm would initially randomly discard half of the agents and then run a  greedy algorithm for monotone submodular objectives.
A first issue is that this approach has not been studied for knapsack constraints. And while it is tempting to believe that random greedy algorithms easily extend to such constraints, it does not seem straightforward (see also Remark \ref{rem:SKS}). 
A second, and probably more serious, issue is that even if a random greedy algorithm directly worked for a knapsack constraint in terms of approximate optimality and truthfulness, budget-feasibility crucially depends on the monotonicity of the objective function \citep{ChenGL11,Singer10}. So, one still needs to deal with the fact that for non-monotone objectives the payments of simple greedy algorithms (like the one by Singer \citep{Singer10}) can be unbounded.

At the heart of our approach lies a novel deterministic greedy algorithm for non-monotone submodular maximization under a knapsack constraint. Our algorithm builds two candidate solutions \emph{simultaneously}, yet prevents agents to jump from one solution to the other by changing their cost. To do the latter we offer each agent a take-it-or-leave-it price based on an estimate of the optimal value which we obtain by sampling. Moreover, this is the first mechanism for the problem where---crucially---the agents are not ordered with respect to their marginal value per cost. This further allows us to appropriately modify the algorithm and adapt it to the online secretary setting and to settings with additional feasibility constraints, while maintaining all its desired properties. 

All of our mechanisms are randomized and, in fact, random sampling is an essential building block in our approach. Obtaining a good estimate of the optimal value via random sampling has been crucial in previous works on  budget-feasible mechanism design for monotone objectives as well \citep{BeiCGL17,BadanidiyuruKS12,ABM17,LeonardiMSZ17}.
Designing \emph{deterministic} budget-feasible mechanisms seems very challenging.  Beyond additive valuation functions \citep{Singer10,ChenGL11}, no deterministic, polynomial-time $O(1)$-approximation mechanisms are known, except for some specific well-behaved objectives \citep{Singer12,ABM16,HorelIM14,DobzinskiPS11,ABM17}. In order to obtain a constant approximation ratio while maintaining truthfulness, one would need to compare the single most valuable agent to an easy-to-calculate estimate of the optimal value that is also non-increasing to each agent's cost. Obtaining deterministic, budget-feasible, $O(1)$-approximation mechanisms is an intriguing topic for future research.

\medskip\noindent 
\textbf{Related Work.}
As mentioned above, the study of budget-feasible mechanisms was
initiated by Singer \citep{Singer10}, who gave a randomized $O(1)$-approximation mechanism for monotone submodular functions. Later, Chen et al.~\citep{ChenGL11} significantly improved the approximation ratio and also suggested a deterministic $O(1)$-approximation mechanism, albeit with superpolynomial running time. Several follow-up results modified this deterministic mechanism so that it runs in polynomial time for special cases, including coverage functions \citep{Singer12,ABM16} and information gain functions \citep{HorelIM14}.    
For subadditive functions, Dobzinski et al.~\citep{DobzinskiPS11} suggested a  $O(\log^2 n)$-approximation mechanism, and  gave the first constant factor mechanisms for  a special case of non-monotone objectives, namely cut functions. The factor for subadditive functions was later improved to $O(\log n / \log \log n)$ by Bei et al.~\citep{BeiCGL17}, who also gave a randomized $O(1)$-approximation mechanism for XOS functions, albeit in exponential time in the value query model, and further initiated the Bayesian analysis in this setting.\footnote{Bei et al.~\shortcite{BeiCGL17} propose an $O(1)$-approximation mechanism for  \emph{non-decreasing} XOS objectives that runs in polynomial time in the much stronger \emph{demand query} model. However, they discuss how to extend their result to general XOS functions via the use of $\hat{v}(S)=\max_{T\subseteq S} v(T)$. It is easy to see that $\hat{v}$ is non-decreasing and that $S$ is an  optimal solution of $v$ if and only if it is a minimal optimal solution for $\hat{v}$. 
Moreover, Gupta et al.~\shortcite{GuptaNS17} proved that if $v$ is general XOS then $\hat{v}$ is monotone XOS.
It should be noted that this transformation does not work for submodular functions \citep{ABM17}. 
Therefore, known results for monotone submodular functions do not extend to the non-monotone case, even in the demand query model.} 
Amanatidis et al.~\citep{ABM17} suggested $O(1)$-approximation mechanisms for a subclass of non-monotone submodular objectives, namely symmetric submodular objectives, however their approach does not seem to generalize beyond this subclass. For settings with additional combinatorial constraints, Amanatidis et al.~\citep{ABM16} and Leonardi et al.~\citep{LeonardiMSZ17} gave $O(1)$-approximation mechanisms for additive valuation functions subject to independent system constraints.
There is also a line of related work under the {\em large market} assumption (where no participant can significantly affect the market outcome), which allows for 
mechanisms with improved performance (see, e.g., \citep{SinglaK13,AnariGN14,GoelNS14,BalkanskiH16,JT17}).

The online version of the problem was introduced and studied by Badanidiyuru et al.~\citep{BadanidiyuruKS12} who give an $O(1)$-approximation mechanism for monotone submodular functions. This is closely related to the purely algorithmic version of the problem (i.e., without the incentives), namely the Submodular Knapsack Secretary Problem introduced by Bateni et al.~\citep{BateniHZ13} as a generalization of the Knapsack Secretary Problem \citep{BabaioffIKK07}. 
Bateni et al.~studied the problem for monotone and non-monotone submodular objectives, although their argument for the latter case is not sound (see Remark \ref{rem:SKS}). While the monotone submodular case has been improved \citep{FeldmanNS11b} and generalized  \citep{KesselheimT17}, there is no follow-up work on the non-monotone case to the best of our knowledge.

On maximization of submodular functions subject to knapsack or other type of constraints, there is a vast literature, going back several decades  (see, e.g., \citep{NemhauserWF78,Wolsey82}). 
Focusing on knapsack constraints, there is a rich line of recent work on developing algorithms on continuous relaxations of the problem (see, e.g., \citep{FeldmanNS11,KulikST13,ChekuriVZ14,EneN17} and references therein) achieving an $e$-approximation for non-monotone objectives.
However, the most relevant recent work to ours is that of Gupta et al.~\citep{GuptaRST10} who  proposed a deterministic $6$-approximation algorithm for the non-monotone case, related on a high level to our main approach. Gupta et al.~also gave algorithms for certain constrained secretary problems, although not with knapsack constraints.  
When $\ell$ knapsack constraints and a $p$-system constraint are both present, the algorithmic state-of-the-art is a $(p+2\ell+1)$-approximation algorithm for the monotone submodular case due to Badanidiyuru and Vondr{\'{a}}k \shortcite{BadanidiyuruV14} and a $(p+1)(2p+2\ell+1) / p$-approximation algorithm for the non-monotone submodular case due to Mirzasoleiman et al.~\shortcite{MirzasoleimanBK16}.

As mentioned above, there is a line of work that uses random greedy algorithms for maximizing non-monotone submodular objectives subject to other combinatorial constraints \citep{FeldmanZ18,CGQ15,BuchbinderFNS14,FeldmanHK17}. Although  not directly related to our work, there are underlying similarities as the algorithms developed are simple, greedy and often extend to online settings. Additionally, if one could resolve the issue of the payments being unbounded, a random greedy version of Singer's mechanism could lead to significantly improved approximation guarantees in our setting.

\section{Preliminaries}
\label{sec:prels}

We  use $A =[n]= \{1, 2, \dots, n\}$ to denote a set of $n$ agents. Each agent $i$ is associated with a private cost $c_i$, denoting the cost for participating in the solution. 
We consider a procurement auction setting, where the auctioneer is equipped with a valuation function $v:2^{A}\to \mathbb{Q}_{\ge 0}$ and a  budget $B>0$. For $S\subseteq A$, $v(S)$ is the value derived by the auctioneer if the set $S$ is selected (for singletons, we will often write $v(i)$ instead of $v(\{i\})$).  Therefore, the algorithmic goal in all the problems we study is to select a set $S$ that maximizes $v(S)$ subject to the constraint $\sum_{i\in S} c_i \leq B$.
We assume oracle access to $v$ via \emph{value queries}, i.e., we assume the existence of a polynomial time value oracle that returns $v(S)$ when given as input a set $S$.

A  function $v$ is \emph{non-decreasing} (often referred to as \emph{monotone}), if $v(S) \le v(T)$ for any $S \subseteq T \subseteq A$. We consider general (i.e., not necessarily monotone), normalized (i.e., $v(\emptyset)=0$), non-negative submodular valuation functions.
Since marginal values are extensively used, we adopt the shortcut $v(i\,|\,S)$ for the marginal value of agent $i$ with respect to the set $S$, i.e.,  $v(i\,|\,S) = v(S \cup \{i\}) - v(S)$.
The following three definitions of submodularity are equivalent. While definition \emph{(i)} is the most standard, the other two alternative definitions will be useful later on.

\begin{definition}\label{def:SM}
	A  function $v$, defined on $2^A$ for some set $A$, is \emph{submodular} if and only if
	\begin{enumerate}[label=\emph{(\roman*)},itemsep=4pt,topsep=4pt]
		\item $v(i\,|\,S) \geq v(i\,|\,T)$ for all $S\subseteq T \subseteq A$, and $i\not\in T$. 
		\item $v(S) + v(T) \ge v(S\cup T) + v(S\cap T)$ for all $S, T \subseteq A$. \label{def:sub2}
		\item $v(T) \le v(S) + \sum_{i \in T \mysetminus S} v(i\,|\,S) - \sum_{i \in S \mysetminus T} v(i\,|\,S\cup T \mysetminus \{i\})$ for all $S, T \subseteq A$. \label{def:sub3}
	\end{enumerate}
\end{definition}

In Section \ref{sec:lower_bounds} we also deal with valuation functions that come from a superclass of submodular functions, namely \emph{XOS}  or \emph{fractionally subadditive} functions. In particular, it is known that non-negative (monotone) submodular functions are a strict subset of (monotone) XOS functions \citep{LehmannLN06,GuptaNS17}. 

\begin{definition}
	A function $v$, defined on $2^A$ for some set $A$, is \emph{XOS} or \emph{fractionally subadditive}, if there exist additive functions $\alpha_1, \ldots,\alpha_r$, for some finite $r$, such that 
	$v(S) = \max_{i\in[r]} \alpha_i(S)$.
\end{definition}

We often need to argue about optimal solutions of sub-instances of the original instance $(A, v, \mathbf{c}, B)$.
Given a cost vector $\mathbf{c}$, and a subset $X\subseteq A$, we denote by $\mathbf{c}_X$ the projection of $\mathbf{c}$ on $X$, and by $\mathbf{c}_{-X}$ the projection of $\mathbf{c}$ on $A\mysetminus X$. 
By $\opt(X, v, \mathbf{c}_X, B)$ we denote 
the value of an optimal solution to the problem restricted 
on $X$.
Similarly, $\opt(X, v,  \infty)$ denotes the value of an optimal solution to the unconstrained version of the problem restricted 
on $X$.
For the sake of readability, we usually drop the valuation function and the cost vector, 
and  write $\opt(X, B)$ and $\opt(X, \infty)$, respectively.

\medskip

\noindent{\bf Mechanism Design.} 
In the strategic version that we consider here, every agent $i\in A$ only has his true cost $c_i$ as private information. Hence, this is a \emph{single-parameter environment}.
A mechanism $\mathcal{M}=(f,p)$ in our context consists of an outcome rule $f$ and a payment rule $p$. Given a vector of cost declarations, 
$\mathbf{b} = (b_i)_{i \in A}$, where $b_i$ denotes the cost reported by agent 
$i$, the outcome rule of the mechanism selects the set $f(\mathbf{b})\subseteq A$. At the same time, it computes payments $p(\mathbf{b}) = (p_i(\mathbf{b}))_{i \in A}$ 
where $p_i(\mathbf{b})$ denotes the payment issued to agent $i$. Hence, the final utility of agent $i$ is $p_i(\mathbf{b}) - c_i$.

Unless stated otherwise,  our mechanisms run in polynomial time in the value query model. Further properties we want to enforce in our mechanism design problem are the following. 
\begin{definition}
	A mechanism $\mathcal{M}=(f,p)$ is 
	\begin{itemize}[itemsep=2pt,topsep=2pt]
		\item \emph{truthful}, if reporting $c_i$ is a dominant strategy for every agent $i$.
		\item \emph{individually rational}, if $p_i(\mathbf{b})\geq 0$ for every $i\in A$, and $p_i(\mathbf{b}) \geq c_i$, for every $i\in f(\mathbf{b})$.
		\item \emph{budget-feasible}, if $\sum_{i\in A} p_i(\mathbf{b}) \leq B$ for every $\mathbf{b}$.
	\end{itemize}
\end{definition}

For our randomized mechanisms we use the strong notion of \textit{universal truthfulness}, which means that the mechanism is a probability distribution over deterministic truthful mechanisms. As all the mechanisms we suggest are universally truthful, we will consistently use  $\mathbf{c} = (c_i)_{i \in A}$ rather than $\mathbf{b} = (b_i)_{i \in A}$ for the declared costs in their description and analysis.

To design truthful mechanisms for single-parameter environments, we  use a characterization by Myerson \shortcite{Myerson81}. 
We say that an outcome rule $f$ is {\em monotone}, if for every agent $i\in A$, and any vector of cost declarations $\mathbf{b}$, if $i\in f(\mathbf{b})$, then $i\in f(b_i', \mathbf{b}_{-i})$ for $b_i' \leq b_i$. That is, if an agent $i$ is selected  
by declaring cost $b_i$, then  he should still be selected by declaring a lower cost.
Myerson's lemma, below, implies that monotone algorithms admit truthful payment schemes 
(often referred to as {\it threshold payments}). 
This greatly simplifies the design of truthful mechanisms,  as one may focus  on constructing monotone algorithms rather than having to worry about the payment scheme. 
For all of our mechanisms, we assume that the underlying payment scheme is given by Myerson's lemma.

\begin{lemma}[\textnormal{ Myerson \citep{Myerson81}}]\label{lem:myerson}
	Given a monotone algorithm $f$, there is a unique payment scheme $p$, such that $(f, p)$ is a truthful and individually rational  mechanism, given by
	\begin{displaymath}
	p_i(\mathbf{b})= \left\{ \begin{array}{ll}
	\sup_{b_i'\in [c_i, \infty)} \{b_i': i\in f(b_i', \mathbf{b}_{-i})\}\,, & \textrm{\emph{ if\ \  }} i\in f(\mathbf{b}),\\
	0\,, & \textrm{\emph{ otherwise.}}
	\end{array} \right.
	\end{displaymath}
\end{lemma}

\begin{remark}\label{rem:bounded_costs}
We may assume, without loss of generality, that in any given instance all the costs are upper bounded by the budget. To see this notice that neither our mechanisms nor the optimal offline solution will ever consider any agent with cost higher than $B$. Furthermore, no agent has an incentive to misreport a very high true cost. Indeed, due to budget-feasibility, if agent $i$ reports a cost $b_i \le B$ instead of his true cost $c_i > B$ and is selected, then he has utility $p_i(\mathbf{b}) - c_i < B-B = 0$.
Thus, in all of our mechanisms we implicitly assume a preprocessing step that removes all the agents with declared costs exceeding $B$. The resulting instance (given as input to the corresponding mechanism) has the same set of optimal solutions subject to the budget constraint as the original one. Note that in the case of the online mechanism \nameref{fig:Main-Submodular-Online} rejecting such agents as they arrive suffices. 
\end{remark}

\begin{remark}\label{rem:tie-breaking}
We should stress that wherever tie-breaking is needed (e.g., in lines \ref{line:SG_argmax} and \ref{line:SG_max_soln} of \nameref{fig:Simultaneous Greedy}, during the execution of the auxiliary algorithms $\alg_1$, $\alg_2$ and $\alg_3$, etc.), we assume the consistent use of a tie-breaking rule that is \emph{independent of the declared costs}. An obvious such choice would be a deterministic lexicographic tie-breaking rule. 
\end{remark}

\section{An Efficient Mechanism for Submodular Objectives}
\label{sec:offline}
The main result of  this section is the first 
$O(1)$-approximation mechanism (termed \nameref{fig:Main-Submodular}  below) for non-monotone submodular valuation functions.

\begin{theorem}\label{thm:Main-Submodular}
	\nameref{fig:Main-Submodular} is a  universally truthful, individually rational, budget-feasible, $O(1)$-approximation mechanism.
\end{theorem}

At the heart of our approach lies a novel greedy algorithm for non-monotone submodular maximization under a knapsack constraint (\nameref{fig:Simultaneous Greedy} below). 
As we mentioned in the \nameref{sec:intro}, all known mechanisms use the same greedy subroutine: sort all agents in decreasing order of marginal value per cost and pick as many agents as possible before hitting some threshold. While for monotone submodular objectives this gives a non-trivial approximation guarantee, for non-monotone objectives may result in arbitrarily bad solutions. 
Moreover, continuous algorithmic approaches for non-monotone submodular maximization under a knapsack constraint \citep{FeldmanNS11,KulikST13} seem very unlikely to yield monotone allocation rules, and thus truthful mechanisms. The only algorithm that is conceptually close to our approach is the two-pass greedy algorithm of Gupta et al.~\shortcite{GuptaRST10}, that runs Sviridenko's greedy algorithm twice and then maximizes without the knapsack constraint to get a deterministic $6$-approximate solution.
The intuition behind this approach is that submodularity prevents the greedy algorithm from getting stuck in consecutive ``bad'' local maxima.
Despite being significantly simpler, however, this two-pass greedy algorithm still suffers irreparably with respect to monotonicity, as it allows agents to jump from one solution to the other by changing their cost.   

Here we introduce \nameref{fig:Simultaneous Greedy}, a greedy mechanism that builds two candidate solutions \emph{simultaneously}. While the analysis of Gupta et al.~\shortcite{GuptaRST10} does not apply here (our solutions are neither built sequentially nor according to the standard greedy algorithm), 
	the way we obtain our approximation guarantee is of the same flavor: at least one of the solutions will  contain an approximately optimal set.
At the same time \nameref{fig:Simultaneous Greedy} prevents agents to choose their favorite candidate solution by misreporting their cost. To achieve that, we offer each agent a take-it-or-leave-it price based on an estimate $x$ of the optimal value which we obtain by sampling.  
This is the first mechanism for the problem where it is crucial that the agents are \emph{not} ordered with respect to their marginal value per cost. This will further allow us to appropriately modify \nameref{fig:Simultaneous Greedy} for the \emph{online} setting of Section \ref{sec:online} while maintaining all its desired properties. 

The parameter $\beta$ is later set to $9.185$ in order to get the approximation factor of Corollary \ref{cor:Main_approximation} but, otherwise, our analysis is independent of its value. $\alg_2$ in line \ref{line:SG_unconstrained} can be any
approximation algorithm for \emph{unconstrained} non-monotone submodular maximization. In particular, here we may use the deterministic 2-approximation algorithm of Buchbinder and Feldman \shortcite{BuchbinderF18}.
\medskip

\begin{algorithm}[H]
	\DontPrintSemicolon 
	\NoCaptionOfAlgo
	\algotitle{\textnormal{\textsc{Simultaneous \allowbreak Greedy}}}{fig:Simultaneous Greedy.title}
	$S_1 =  S_2 = \emptyset$; $B_1 =  B_2 = B$; $U=D$   \tcc*{{\footnotesize each $S_j$ has its own budget $B_j$}}
	\While{$\max_{i\in U, j\in\{1,2\}} v(i|S_j) > 0$}{
		Let $(\hat{\imath},\hat{\jmath}) \in \argmax_{i\in U, j\in\{1,2\}} v(i|S_j)$ \label{line:SG_argmax}\;
		\If{$c_{\hat{\imath}}\le \frac{\beta  B}{x} v(\hat{\imath}|S_{\hat{\jmath}}) \le B_{\hat{\jmath}}$ \label{line:SG_condition}}{
			$S_{\hat{\jmath}} = S_{\hat{\jmath}} \cup \{\hat{\imath}\}$ \;
			$B_{\hat{\jmath}} = B_{\hat{\jmath}} - \frac{\beta  B}{x} v(\hat{\imath}|S_{\hat{\jmath}})$\;}
		$U = U\mysetminus\{\hat{\imath}\}$ \label{line:SG_rejection}\;
	}
	\For{$j\in \{1,2\}$}{
		$T_j=\alg_2(S_j)$ \tcc*{{\footnotesize a 2-approximate solution with respect to  $\opt(S_j, v, \mathbf{c}_{S_j}, \infty)$}} \label{line:SG_unconstrained}
	}
	Let $S$ be the best solution among $S_1, S_2, T_1, T_2$ \label{line:SG_max_soln}\;
	\Return $S$
	\caption{\textsc{Simultaneous Greedy}$(D, v, \mathbf{c}_D, B, x)$} \label{fig:Simultaneous Greedy} 
\end{algorithm}\medskip

Ideally, we would like the rate parameter $x$ to be close to $\opt(A, B)$ and also to be robust in the sense that no single agent can significantly affect its value. 
To achieve that, \nameref{fig:Sample-then-Greedy} randomly partitions the set of agents into two sets $A_1$ and $A_2$, then approximately solves the problem on $A_1$ to obtain an estimate of $\opt(A_1, B)$, and finally uses this $x$ to set the threshold rate for \nameref{fig:Simultaneous Greedy} on $A_2$. 

$\alg_1$ in line \ref{line:StG_submod_knapsack} can be any 
approximation algorithm for non-monotone submodular maximization subject to a knapsack constraint. In particular, here we may use the $e$-approximation algorithm of Kulik et al.~\shortcite{KulikST13} (also see Remark \ref{rem:rand_approx_ratio}). \medskip

\begin{algorithm}[H]
	\DontPrintSemicolon 
	\NoCaptionOfAlgo
	\algotitle{\textnormal{\textsc{Sample-then-Greedy}}}{fig:Sample-then-Greedy.title}
	Put each agent of $A$ in either $A_1$ or $A_2$ independently at random with probability $\frac{1}{2}$ \;
	$x=v(\alg_1(A_1))$ \tcc*{{\footnotesize an $e$-approximation of $\opt(A_1, v, \mathbf{c}_{A_1}, B)$}} \label{line:StG_submod_knapsack}
	\Return \text{\nameref{fig:Simultaneous Greedy}}$(A_2, v, \mathbf{c}_{A_2}, B, x)$ 
	\caption{\textsc{Sample-then-Greedy}$(A, v, \mathbf{c}, B)$} \label{fig:Sample-then-Greedy} 
\end{algorithm}\medskip

\noindent Lemma \ref{lem:Bei-Leonardi} in Subsection \ref{subsec:Main-Submodular}, due to Bei et al.~\shortcite{BeiCGL12} and Leonardi et al.~\shortcite{LeonardiMSZ16}, guarantees that with high probability both $A_1$ and $A_2$ contain enough value subject to the budget constraint for things to work, as long as no agent is too valuable. The latter leads to the final mechanism \nameref{fig:Main-Submodular} (\underline{\textsc{Gen}}eral \underline{\textsc{S}}ub\underline{\textsc{m}}odular-\underline{\textsc{Main}}) that randomizes between all the above and just returning a best singleton.\medskip

\begin{algorithm}[H]
	\DontPrintSemicolon 
	\NoCaptionOfAlgo
	\algotitle{\textnormal{\textsc{GenSm-Main}}}{fig:Main-Submodular.title}
		\emph{With probability} $p=0.201$ \textbf{:} \Return $i^*\in \argmax_{i\in A} v(i)$ \;
		\emph{With probability} $1-p$ \textbf{:} \Return \nameref{fig:Sample-then-Greedy}$(A, v, \mathbf{c}, B)$ \;
	\caption{\textsc{GenSm-Main}$(A, v, \mathbf{c}, B)$} \label{fig:Main-Submodular} 
\end{algorithm}

\begin{remark}\label{rem:7-approx}
Here it is necessary that \nameref{fig:Simultaneous Greedy} uses thresholds in order to achieve the properties stated in Theorem \ref{thm:Main-Submodular}. While this goes beyond the point of this work, one can follow the same approach of greedily building two solutions at the same time (using a variant of Sviridenko's algorithm \shortcite{Sviridenko04}), in order to design a deterministic 7-approximation  algorithm for maximizing non-monotone submodular functions subject to a knapsack constraint. 
\end{remark}

\subsection{Proving the Properties of \nameref{fig:Main-Submodular}}
\label{subsec:Main-Submodular}
We fix some additional notation to facilitate the presentation of the proofs. We use $(D, v, \mathbf{c}, B, x)$ for a generic instance given to \nameref{fig:Simultaneous Greedy} and $S$ for the set returned. By $i_1, i_2, \ldots, i_{t}$ we denote the sequence of agents of $D$ examined during this execution of the algorithm in this exact order. All the agents of $S$ clearly appear within this sequence, so for any particular $\ell\in S$ we have that $\ell=i_k$ for some $k$. By $j_k$ we denote the index $\hat{\jmath}$ picked during the $k$th execution of line \ref{line:SG_argmax} of \nameref{fig:Simultaneous Greedy}, while by $S_{j_k}^{(k)}$ and $B_{j_k}^{(k)}$ we denote the set $S_{j_k}$ and its remaining budget, respectively, at that time.
Conventionally, we  use notation like $S_{j_k}^{(k+1)}$ to denote $S_{j_k}$ right after the $k$th execution of line \ref{line:SG_rejection}, even if line \ref{line:SG_argmax} is never executed more than $k$ times.
Recall that we use a 
tie-breaking rule that is independent of the costs, as mentioned in Remark \ref{rem:tie-breaking}.

\begin{lemma}\label{lem:SG_monotone}
The allocation rule defined by \nameref{fig:Simultaneous Greedy} is monotone. Thus, using the threshold payments of Myerson's lemma, the resulting mechanism is truthful and individually rational. 
\end{lemma}

\begin{proof}
By Lemma \ref{lem:myerson}, we just need to show that the allocation rule is monotone, i.e., a winning agent remains a winner if he decreases his cost. In fact, we show something stronger, namely that no winning agent can affect the output of \nameref{fig:Simultaneous Greedy} by lowering his bid. 

Let $S$ be the set returned when the input is $(D, v, \mathbf{c}, B, x)$ and fix some agent $i_k\in S$. That is, during the $k$th execution of line \ref{line:SG_argmax}, $(\hat{\imath},\hat{\jmath})=(i_k, j_k)$.
Fix the vector $\mathbf{c}_{-i_k}$ for the other agents, and suppose that agent $i_k$ declares $c_{i_k}'<c_{i_k}$. Clearly, the execution of \nameref{fig:Simultaneous Greedy}$(D, v, (\mathbf{c}_{-i_k},c_{i_k}'), \allowbreak B, x)$ will be exactly the same as before for agents  $i_1, \ldots, i_{k-1}$. Further, $\hat{\jmath}$ will again be $j_k$. Thus, $i_k$ will again be added to the $S_{j_k}^{(k)}$ since 
\[
c_{i_k}' < c_{i_k} \le \frac{\beta  B}{x} v\Big(i_k \, \big|\, S_{j_k}^{(k)}\Big) \le B_{j_k}.
\] 
After updating $B_{j_k}$ to $B_{j_k}^{(k)} - \frac{\beta  B}{x} v\big(i_k|S_{j_k}^{(k)}\big)$, everything is exactly the same as in the beginning of the $(k+1)$th iteration of the original execution of \nameref{fig:Simultaneous Greedy}$(D, v, \mathbf{c}, B, x)$ and, therefore, the algorithm will proceed in exactly the same way to produce the same output $S$. In particular, agent $i_k$ will still be a winner. 
\end{proof}

In all the following statements, when we refer to mechanisms, we always assume threshold payments.
Before we study the total payment, we should point out that enforcing budget-feasibility has been the main source of technical difficulties in the budget-feasible mechanism design literature.  A significant advantage of the take-it-or-leave-it approach used in threshold mechanisms like \nameref{fig:Simultaneous Greedy} is that the budget-feasibility becomes much more manageable. To some extent this comes at the expense of the approximation guarantee and its analysis, but also offers some additional flexibility that will be explored in Sections \ref{sec:online} and \ref{sec:combinatorial}. 

\begin{lemma}\label{lem:SG_budget-feasible}
	The mechanism \nameref{fig:Simultaneous Greedy} is budget-feasible.
\end{lemma}

\begin{proof}
Let $S$ be the set returned given the instance $(D, v, \mathbf{c}, B, x)$ and fix $i_k\in S$. We claim that the payment $p_{i_k}(\mathbf{c})$ is exactly $\pi_k = \frac{\beta  B}{x} v\big(i_k|S_{j_k}^{(k)}\big)$, i.e., $i_k\in S$ if and only if he bids $c_{i_k}'\le\pi_k$. First note that $i_k$ cannot affect the time when he is examined by the mechanism or which agents come before him. So, since $\mathbf{c}_{-i_k}$ is fixed, during the $k$th execution of line \ref{line:SG_argmax}, he is always ``offered'' $\pi_k$; either he accepts, i.e., $c_{i_k}'\le\pi_k$, and the algorithm proceeds in the exact same way as with $c_{i_k}'=c_{i_k}$ (see also the proof of Lemma \ref{lem:SG_monotone}) or he rejects, i.e., $c_{i_k}'>\pi_k$, and he is removed from the active set of agents. Once an agent is removed, however, he is never reexamined and thus, if $c_{i_k}'>\pi_k$ then $i_k$ is not in the winning set.

Recall that $S$ can be any of $S_1, S_2, T_1, T_2$. We will show that all four sets are budget-feasible. Let $T_1= \{i_{a_1}, i_{a_2}, \ldots, i_{a_{|T_1|}}\}$ and $S_1= \{i_{b_1}, i_{b_2}, \ldots, i_{b_{|S_1|}}\}$,
where $(a_i)_{i=1}^{|T_1|}$ is a subsequence of $(b_i)_{i=1}^{|S_1|}$ which is a subsequence of $1, 2, \ldots, t$. Recall that the budget $B_1$ for $S_1$ is never exhausted. We have
\[\sum_{\tau=1}^{|T_1|} \pi_{a_{\tau}} \le \sum_{\tau=1}^{|S_1|} \pi_{b_{\tau}} = \sum_{\tau=1}^{|S_1|} \frac{\beta  B}{x} v\left(i_{b_{\tau}}|S_{j_{b_{\tau}}}^{({b_{\tau}})}\right) = B - B_1^{(|S_1|+1)} \le B \,. \]
The first and the second sum represent the total payment when $S=T_1$ and when $S=S_1$, respectively. The budget-feasibility of $T_2$ and $S_2$ is proved in the exact same way.
\end{proof}

\begin{corollary}\label{cor:Main_properties}
The mechanism \nameref{fig:Main-Submodular} is universally truthful, individually rational and budget-feasible.
\end{corollary}

\begin{proof}
Given that $\alg_1$ and $\alg_2$ run in polynomial time, and that it is straightforward to determine the payments (Lemma \ref{lem:SG_budget-feasible}), it is clear that \nameref{fig:Main-Submodular} is a polynomial-time mechanism.

Further, \nameref{fig:Main-Submodular} is a probability distribution over the mechanism that returns $i^*\in \linebreak \argmax_{i\in A} v(i)$ and 
\nameref{fig:Simultaneous Greedy}$(D, v, \mathbf{c}_{D}, B, v(\alg_1(A\mysetminus D)))$ for all $D\subseteq A$.
The simple mechanism that returns $i^*$ and pays him the threshold payment is truthful and individually rational. Further, it is clear that the threshold payment is exactly $B$, so this mechanism is budget-feasible as well.

The desired properties of \nameref{fig:Main-Submodular} now follow from Lemmata \ref{lem:SG_monotone} and \ref{lem:SG_budget-feasible} and the above observations. 
\end{proof}

\begin{lemma}\label{lem:SG_approximation}
If there is a positive integer $\ell$ such that  $\max_{i\in D} v(i) < \frac{x}{\ell\cdot\beta}$, then
\nameref{fig:Simultaneous Greedy}$(D, v, \mathbf{c}, B, x)$ outputs a set $S$ such that 
\[v(S)\ge \min\left\lbrace \frac{\ell x}{(\ell+1)\beta} \,,\,  \frac{1}{6} \left( \opt(D, B) - \frac{2x}{\beta}\right)  \right\rbrace  \,.\]
\end{lemma}

\begin{proof}
Let $t$ be the number of times line \ref{line:SG_argmax} was executed. At the end of the $t$th iteration, $U$ is the set of agents never examined. That is, $U$ only contains agents that have non-positive marginal utilities with respect to $S_{1}^{(t+1)}$ and $S_{2}^{(t+1)}$. For the sake of readability, we henceforth  use $S_{1}$ and $S_{2}$ to denote $S_{1}^{(t+1)}$ and $S_{2}^{(t+1)}$, respectively. Let $R=D\mysetminus (U\cup S_1\cup S_2)$ be the agents $i_k$ that were considered at some point by the mechanism but were \emph{rejected}, i.e., not added to either $S^{(k)}_1$ or $S^{(k)}_2$. 
We first partition $R$ into two sets depending on \emph{why} the corresponding agents were rejected. 
The set 
\[
R_{\mathbf{c}} = \bigg\{ i_k \,\Big |\, \frac{\beta  B}{x} v\Big( i_k \, \big | \, S_{j_k}^{(k)}\Big)  < c_{\hat{\imath}} \bigg\} 
\]
contains the agents rejected because the first inequality in line \ref{line:SG_condition} was violated during the corresponding iteration. 
Similarly, the set 
\[
R_B  = \bigg\lbrace i_k \,\Big |\, B_{j_k}^{(k)} < \frac{\beta  B}{x} v\Big( i_k\, \big|\, S_{j_k}^{(k)}\Big) \bigg\rbrace
\]
contains the agents rejected because the  second inequality in line \ref{line:SG_condition} was violated. 
Clearly, $R = R_{\mathbf{c}} \cup R_B$. We consider two cases, depending on whether $R_B$ is empty or not. \smallskip

\noindent \textbf{Case 1.} Assume that  $R_B\neq \emptyset$ and let $i_k\in R_B$. That is, during the $k$th execution of line \ref{line:SG_argmax}, $(\hat{\imath},\hat{\jmath})=(i_k, j_k)$, but $\frac{\beta  B}{x} v\big( i_k|S^{(k)}_{j_k}\big)  > B^{(k)}_{j_k}$. Let $S^{(k)}_{j_k}= \{i_{a_1}, i_{a_2}, \ldots, i_{a_{s}}\}$, where $(a_i)_{i=1}^{s}$ is a subsequence of $1, 2, \ldots, t$. Further, notice that, by its definition, $B^{(k)}_{j_k} = B - \sum_{\tau=1}^{s} \frac{\beta B}{x} v\big( i_{a_{\tau}} |S^{(a_{\tau})}_{j_{k}}\big)$.
We have
\begin{IEEEeqnarray}{rCl}\label{eq:SG_approximation_1}
 v\left( S^{(k)}_{j_k}\right)& = &\sum_{i=1}^{s} v\left( i_{a_i} |S^{(a_i)}_{j_{k}}\right) 
 =  \frac{x}{\beta B} \left( B- B^{(k)}_{j_k}\right) \nonumber\\
	& >  &   \frac{x}{\beta  B} \left( B- B^{(k)}_{j_k}\right) + \frac{x}{\beta  B} B^{(k)}_{j_k} - v\left( i_k|S^{(k)}_{j_k}\right) 
	 =  \frac{x}{\beta}  -  v\left( i_k|S^{(k)}_{j_k}\right) \,. 
\end{IEEEeqnarray}
By submodularity and the way the agents in $S^{(k)}_{j_k}$ are chosen, we have 
\[v(i_{a_1}) = v\left( i_{a_1} |S^{(a_1)}_{j_k}\right) \ge v\left( i_{a_2} |S^{(a_2)}_{j_k}\right) \ge \ldots \ge v\left( i_{a_s} |S^{(a_s)}_{j_k}\right) \ge v\left( i_{k} |S^{(k)}_{j_k}\right) \,.\]
Yet, each one of these values is at most $\max_{i\in D} v(i) < \frac{x}{\ell\cdot\beta}$.
Combining with \eqref{eq:SG_approximation_1}, we have
\begin{equation*}
	\ell\cdot\max_{i\in D} v(i) < \frac{x}{\beta}  \le  v\left( S^{(k)}_{j_k}\right) + v\left( i_k|S^{(k)}_{j_k}\right) 
	\le  \sum_{\tau=1}^{s} v\left( i_{a_{\tau}} |S^{(a_{\tau})}_{j_{k}}\right)  + v\left( i_k|S^{(k)}_{j_k}\right) 
	\le  \left( s + 1 \right) \cdot v\left( i_{a_1} \right) \,, 
\end{equation*}
and therefore, we conclude that $\big| S^{(k)}_{j_k} \big| = s \ge \ell$. Now we repeat the same argument for the average marginal value in the sum $\sum_{\tau=1}^{s} v\big( i_{a_{\tau}} |S^{(a_{\tau})}_{j_{k}}\big)$. 
Using the simple observation that the smallest term of a sum cannot exceed the average of the remaining terms, we get
\begin{equation}\label{eq:SG_approximation_2}
\frac{x}{\beta} \le v\left( S^{(k)}_{j_k}\right) + v\left( i_k|S^{(k)}_{j_k}\right) \le v\left( S^{(k)}_{j_k}\right) + \frac{1}{s} \sum_{\tau=1}^{s} v\left( i_{a_{\tau}} |S^{(a_{\tau})}_{j_{k}}\right) \le \frac{s + 1}{s} v\left( S^{(k)}_{j_k}\right) \le \frac{\ell +1}{\ell} v\left( S^{(k)}_{j_k}\right) \,,
\end{equation}
where the last inequality follows from the fact that $f(z)=\frac{z+1}{z}$ is decreasing.

Finally, to get the approximation guarantee for this case, we combine \eqref{eq:SG_approximation_2} with the fact that $S$ is at least as good as each greedy solution: 
\[v(S)  \ge  v(S_{j_k}) \ge v\left( S^{(k)}_{j_k}\right) \ge \frac{\ell}{\ell +1}\cdot \frac{x}{\beta} \,.\]

\noindent \textbf{Case 2.} Now assume that  $R_B = \emptyset$, i.e., $R=R_{\mathbf{c}}$. Let $C^*$ be an optimal solution for the given instance and define $C_1= C^* \cap S_1$, $C_2= C^* \cap S_2$ and $C_3 = C^* \mysetminus (C_1\cup C_2)$. By subadditivity, we have 
\begin{equation}\label{eq:SG_approximation_3}
\opt(D, B) = v(C^*) \le v(C_1) + v(C_2) + v(C_3) \,.
\end{equation}
Recall that $T_j=\alg_2(S_j), j\in\{1,2\}$, is a 2-approximate solution with respect to  $\opt(S_j, \infty)$.
Thus, $v(C_j) \le \opt(S_j,  B) \le 2\cdot v(T_j)$, for $j\in\{1,2\}$, and inequality \eqref{eq:SG_approximation_3} gives
\begin{equation}\label{eq:SG_approximation_4}
\opt(D, B) \le 2v(T_1) + 2v(T_2) + v(C_3) \,.
\end{equation}
Upper bounding $v(C_3)$ in terms of $S_1, S_2, T_1, T_2$ and $r$ is somewhat more involved. We begin by invoking the non-negativity of $v$, as well as its submodularity (as defined in of Definition \ref{def:SM}\ref{def:sub2}) on $S_1\cup C_3$ and $S_2\cup C_3$. We have
\begin{equation}\label{eq:SG_approximation_5}
v(C_3) \le v(C_3) + v(S_1\cup C_3\cup S_2) \le v(S_1\cup C_3) + v(S_2\cup C_3)  \,.
\end{equation}
In order to upper bound $v(S_1\cup C_3)$ we again use the submodularity of $v$, together with a couple of facts about the marginal utilities of agents outside of $S_1$. Since the mechanism stopped after $t$ iterations, $\max_{i\in D\mysetminus(S_1\cup S_2\cup R)} v(i|S_1) \le 0$. 
Also, given that $R=R_{\mathbf{c}}$, for all agents that got rejected at some point, we  know that they had very low marginal value per cost ratio with respect to both $S_1$ and $S_2$. 
In particular, if $i_k\in R$, then $c_{i_k}> \frac{\beta  B}{x} v\big( i_k|S^{(k)}_j\big)$, for both $j\in\{1,2\}$. We may now rely on  Definition \ref{def:SM}\ref{def:sub3} to get
\begin{IEEEeqnarray*}{rCl's}
	v(S_1\cup C_3)& \le  & v(S_1) + \sum_{i_k\in C_3} v( i_{k} |S_1) & \\ 
	& \le & v(S_1) + \sum_{i_k\in C_3\cap R} v( i_{k} |S_1) & {\small (\ $v(i_k|S_1) \le 0$ for $i_k\in C_3\mysetminus R $\ )}  \\
	& \le  &  v(S_1) +  \sum_{i_k\in C_3\cap R} v\left( i_{k} |S^{(k)}_1\right) & {\small (\ by submodularity, $v( i_{k} |S_1)\le  v\left( i_{k} |S^{(k)}_1\right)$ for  $i_k\in D$\ )} \\ 
	& \le &  v(S_1) +  \sum_{i_k\in C_3\cap R} \frac{x}{\beta  B} c_{i_k}  &  {\small (\ $\frac{\beta  B}{x} v\big( i_k|S^{(k)}_1\big) <  c_{i_k}$ for $i_k\in R$\ )}
\end{IEEEeqnarray*}
Similarly, $v(S_2\cup C_3)\le   v(S_2) +  \sum_{i_k\in C_3\cap R} \frac{x}{\beta  B} c_{i_k}$.
Also, recall that $\sum_{i \in C^*} c_i \le B$ to get 
\begin{equation}\label{eq:SG_approximation_6}
v(S_j\cup C_3)\le   v(S_j) +  \frac{x}{\beta}\,, \text{\ \ for \ \ } j\in\{1,2\} \,.
\end{equation}
Finally, we may combine \eqref{eq:SG_approximation_4}, \eqref{eq:SG_approximation_5} and \eqref{eq:SG_approximation_6} to get 
\begin{IEEEeqnarray*}{rCl}
\opt(D, B) & \le & 2v(T_1) + 2v(T_2) + v(S_1\cup C_3) + v(S_2\cup C_3) \\
	& \le  &  2v(T_1) + 2v(T_2) + v(S_1) + v(S_2) +  \frac{2x}{\beta} 
	\le   6\cdot v(S) +  \frac{2x}{\beta} \,, 
\end{IEEEeqnarray*}
or, equivalently, 
$v(S)  \ge  \frac{1}{6} \big( \opt(D, B) - \frac{2x}{\beta}\big)$. 

Combining Case 1 and Case 2, we obtain the claimed guarantee.
\end{proof}
So far, unless $x=\Theta(\opt(D, B))$, the approximation guarantee seems to be rather weak. In fact, the way \nameref{fig:Simultaneous Greedy} is used within \nameref{fig:Sample-then-Greedy} requires that both $x=v(\alg_1(A_1))$ and $\opt(A_2,  B)$ are $\Theta(\opt(A, B))$. The next technical lemma guarantees that this happens with high probability, unless there is an extremely valuable agent.

\begin{lemma}[\textnormal{Follows from Bei et al.~\citep{BeiCGL17} and Leonardi et al.~\citep{LeonardiMSZ16}}]\label{lem:Bei-Leonardi}
	Consider any submodular function $v(\cdot)$.  For any given subset $T \subseteq A$ and a positive integer $k$ assume that $v(T) \geq k \cdot \max_{i\in T} v(i)$. Further, suppose that $T$ is divided uniformly at random into two subsets $T_1$ and $T_2$. Then with probability at least $\frac{1}{2}$, we have that $v(T_1) \geq \frac{k-1}{4k} v(T)$ and $v(T_2) \geq \frac{k-1}{4k}v(T)$.
\end{lemma}

We are now ready to lower bound the approximation guarantee of \nameref{fig:Sample-then-Greedy}. 

\begin{lemma}\label{lem:StG_approximation}
Assume that for some positive integer $k$, $\opt(A, B) > k \cdot \max_{i\in A}v(i)$. Then with probability at least $\frac{1}{2}$, \nameref{fig:Sample-then-Greedy}$(A, v, \mathbf{c}, B)$ outputs a set $S$ such that 
\[v(S)\ge \min\left\lbrace \frac{\big\lfloor \frac{k-1}{4 e \beta} \big\rfloor  (k-1)}{e  \left( \big\lfloor \frac{k-1}{4 e \beta} \big\rfloor + 1\right) } \,,\,  \frac{\beta (k-1)- 8k}{6}   \right\rbrace \cdot \frac{1}{4  \beta k} \cdot \opt(A, B)  \,.\]
\end{lemma}

\begin{proof}
Let $C^*$ be an optimal solution for the given instance. By applying Lemma \ref{lem:Bei-Leonardi} with $T = C^*$ we have that with probability at least $\frac12$ it holds that $v(A_i\cap C^*) \geq \frac{k-1}{4k} v(C^*)$ for both $i\in \{1,2\}$. In what follows we assume that this is indeed the case.
Thus, 
\[\opt(A, B) \ge x=v(\alg_1(A_1))\ge \frac{1}{e} \opt(A_1,  B) \ge \frac{k-1}{4ek} \opt(A, B)\]
and also $\opt(A_2,  B) \ge \frac{k-1}{4k} \opt(A, B)$.\footnote{For the sake of presentation, here we write $x=v(\alg_1(A_1))\ge \frac{1}{e} \opt(A_1,  B)$ rather than the technically correct $\mathbb{E}(x)\ge \frac{1}{e} \opt(A_1,  B)$. However, as discussed in Remark \ref{rem:rand_approx_ratio}, one can formally deal with this issue with a negligible effect on the expected approximation guarantee while keeping the running time polynomial.}

The lower  bound on $x$ paired with the upper bound on $\max_{i\in A}v(i)$, imply that
\[\max_{i\in A}v(i) < \frac{1}{k} \cdot \opt(A, B) \leq \frac{1}{k}\cdot\frac{4ek\beta}{k-1}\cdot \frac{x}{\beta} \leq \frac{1}{\big\lfloor \frac{k-1}{4 e \beta} \big\rfloor} \cdot \frac{x}{\beta} \,. \]

Thus, we can use Lemma \ref{lem:SG_approximation} with $D = A_2$, $x=v(\alg_1(A_1))$ and $\ell=\big\lfloor \frac{k-1}{4 e \beta} \big\rfloor$. Therefore, 
\nameref{fig:Simultaneous Greedy}$(A_2, v, \mathbf{c}_{A_2}, \allowbreak B, x)$ outputs an $S$ such that 
\begin{IEEEeqnarray*}{+rCl+x*}
	v(S) & \ge & \min\left\lbrace \frac{\big\lfloor \frac{k-1}{4 e \beta} \big\rfloor  x}{\left( \big\lfloor \frac{k-1}{4 e \beta} \big\rfloor + 1\right)\beta} \,,\,  \frac{1}{6} \left( \opt(A_2,  B) - \frac{2x}{\beta}\right)  \right\rbrace  \\
	& \ge  &  \min\left\lbrace \frac{\big\lfloor \frac{k-1}{4 e \beta} \big\rfloor (k-1)}{4ek\left( \big\lfloor \frac{k-1}{4 e \beta} \big\rfloor + 1\right) \beta}  \opt(A, B) \,,\,  \frac{1}{6} \left( \frac{k-1}{4k} \opt(A, B) - \frac{2}{\beta}  \opt(A, B)\right)  \right\rbrace  \\
	& \ge  &  \min\left\lbrace \frac{\big\lfloor \frac{k-1}{4 e \beta} \big\rfloor (k-1)}{e  \left( \big\lfloor \frac{k-1}{4 e \beta} \big\rfloor + 1\right) } \,,\,  \frac{\beta (k-1)- 8k}{6}   \right\rbrace \cdot \frac{1}{4  \beta k} \cdot \opt(A, B) \,. & \qedhere 
\end{IEEEeqnarray*}
\end{proof}
\vspace{1pt}

\begin{corollary}\label{cor:Main_approximation}
The set $S$ returned by \nameref{fig:Main-Submodular}$(A, v, \mathbf{c}, B)$ satisfies 
\[ 505 \cdot \mathbb{E}(v(S)) \ge \opt(A, B)  \,.\]
\end{corollary}

\begin{proof}
Suppose that $\max_{i\in A}v(i) \ge \frac{1}{101} \cdot \opt(A, B)$. Then, with probability $p$ at least $1/101$ of the optimal value is returned. Hence,
\[\mathbb{E}(v(S)) \ge p\cdot \max_{i\in A}v(i) \ge \frac{0.201}{101} \cdot \opt(A, B) \ge \frac{1}{505} \cdot \opt(A, B) \,. \]

Next suppose that $\max_{i\in A}v(i) < \frac{1}{101} \cdot \opt(A, B)$. We may apply Lemma \ref{lem:StG_approximation} with $k=101$. As discussed before the description of mechanism \nameref{fig:Simultaneous Greedy}, the parameter $\beta$ is is equal to $9.185$. This implies that $\big\lfloor \frac{k-1}{4 e \beta} \big\rfloor =1$. By substituting the values of $k$ and $\beta$ to the bound of  Lemma \ref{lem:StG_approximation}, we get that with probability at least $(1-p)/2$ 
\[v(S) \ge   \min\left\lbrace \frac{50}{e } \,,\,  \frac{110.5}{6}   \right\rbrace \cdot \frac{1}{3710.74} \cdot \opt(A, B) \ge    \frac{1}{201.7367} \cdot \opt(A, B)  \,,\]
and thus,
\[\mathbb{E}(v(S)) \ge \frac{1-0.201}{2}\cdot\frac{1}{201.7367} \cdot \opt(A, B) \ge \frac{1}{505} \cdot \opt(A, B) \,.\qedhere \]
\end{proof}

\smallskip\begin{remark}\label{rem:rand_approx_ratio}
	In our mechanisms  we often use randomized approximation algorithms as subroutines. In particular, $\alg_1$ in \nameref{fig:Sample-then-Greedy} and \nameref{fig:Main-Submodular-Online}, $\alg_3$ in \nameref{fig:Main-Constrained}, and $\alg_4$ in \nameref{fig:nMain-Constrained} are all randomized. Yet, in the description of our mechanisms---and more importantly in our analyses---we treat them as if there were deterministic, e.g., we assume that $x$ in line \ref{line:StG_submod_knapsack} of \nameref{fig:Sample-then-Greedy} is at least $\frac{1}{e} \cdot \opt(A_1, v, \mathbf{c}_{A_1}, B)$. While this is not technically accurate, we do so for the sake of presentation. However, we may instead use the fact that for any constants $\delta, \eta >0$, $\alg_1$ can be modified so that \emph{with probability at least $1-\delta$ it returns a solution of value at least $\big(\frac{1}{e}-\eta\big) \cdot \opt(A_1, v, \mathbf{c}_{A_1}, B)$ in polynomial time}, using standard arguments.\footnote{For instance, even without going into the mechanics of $\alg_1$, running the algorithm 
		$\frac{2\log_2{(1-\delta)}}{\log_2{(e-1)} - \log_2{((1+\eta)e -1)}}$ times and keeping the best solution suffices. } Thus, for any constant $\varepsilon>0$ the analysis of \nameref{fig:Main-Submodular} can be adjusted to hold for an approximation factor $505+\varepsilon$ instead. Similarly, the analyses of our other mechanisms can be adjusted accordingly. Actually, since there is some slack in the approximation ratios derived in this work, this extra $\varepsilon$ due to the  randomized subroutines can, in fact, be ``hidden'' in the current ratios.
\end{remark}

\section{Online Procurement}
\label{sec:online}
Note that the mechanism presented in the last section already bares some resemblance to online algorithms for the \emph{secretary model} (although truthfulness is rarely a requirement there). Namely, a part of the input is only used to estimate the quality of the optimal solution and then, based on that estimation, some threshold is set for the remaining instance. On a high level, this is straightforward to adjust for the secretary model; we use the first (roughly) half of the stream of agents to find an estimate of $\opt(A, B)$ and then set a threshold similar to the one in \nameref{fig:Simultaneous Greedy}. However, there are a few issues one has to deal with.

First, \nameref{fig:Simultaneous Greedy} goes through the agents in a specific order (in decreasing order of the maximum marginal value with respect to either one of the two constructed sets). Even though this fact is indeed used in the proof of Lemma \ref{lem:SG_approximation}, we show that even examining agents in arbitrary order works well, albeit with a somewhat worse approximation factor. Note that this is not true when there are other constraints on top of the budget-feasibility requirement, as in Section \ref{sec:combinatorial}.

Second, towards the end, in line \ref{line:SG_unconstrained}, \nameref{fig:Simultaneous Greedy} runs an unconstrained submodular maximization algorithm on $S_1$ and $S_2$ to possibly reveal a subset of them with much higher value. While this is a critical step, we rely on a very elegant result of Feige et al.~\shortcite{FeigeMV11}: a uniformly random set gives a 4-approximation for the unconstrained problem. Thus, every agent that passes the threshold and is added to $S_j$ is only accepted to $T_j$ with probability $1/2$. The actual output of the mechanism is a random choice $S$ between $S_1$, $S_2$, $T_1$ and $T_2$, made \emph{before} the arrival of the first agent. So, while the four sets are built obliviously with respect to the choice of $S$, the agents added to $S$ are irrevocably chosen while everyone else is irrevocably discarded.

One last issue is that we want the mechanism to occasionally return the single most valuable agent. This, however, is easily resolved by running \emph{Dynkin's algorithm} \citep{Dynkin1963} with constant probability instead. 
This mechanism samples the first $n/e$ agents and then it picks the first agent $i'$, among the remaining agents, who is at least as good as the best agent in the sample, i.e., $v(i')\ge\max_{k\le n/e}v(i_k)$. This guarantees that $\mathbb{E}(v(i'))\ge\frac{1}{e} \max_{i\in A}v(i)$, where the expectation is over the order of the agents.

The mechanism \nameref{fig:Main-Submodular-Online} below incorporates all these adjustments, yet maintains all the good properties of \nameref{fig:Main-Submodular}.  
We assume a secretary setting, where the agents arrive \emph{uniformly at random}. In particular, agents have no control over their arrival time, so this is still a single-parameter environment and truthfulness still means \emph{universal truthfulness}, i.e., if we fix the random bits of the mechanism, then \emph{for any arrival order} no agent has an incentive to lie.  
Moreover, note that \nameref{fig:Main-Submodular-Online} is \emph{order oblivious}. That is, Theorem \ref{thm:Main-Submodular-Online} holds even when the order of samples and the order of values are (separately) chosen by an adversary.

Again, $\alg_1$ is the $e$-approximation algorithm of Kulik et al.~\shortcite{KulikST13}. The parameter $\beta$ is set to $8.725$ and, like the parameter in \nameref{fig:Simultaneous Greedy}, is only relevant for the approximation factor. 

\begin{theorem}\label{thm:Main-Submodular-Online}
	\nameref{fig:Main-Submodular-Online} is a  universally truthful, individually rational, budget-feasible online mechanism  
	and achieves an $O(1)$-approximation in the secretary model.
\end{theorem}

\begin{proof}
	Fix any particular arrival order $i_1, i_2, \ldots, i_n$ of the agents.	
	
	By fixing the sequence $\rho$ of the random bits of the mechanism, we get a deterministic allocation rule $\nameref{fig:Main-Submodular-Online}(\rho)$. 
	In the case where this is Dynkin's algorithm, it is straightforward that---coupled with the threshold payment of $B$ to the possible winner---it is truthful, individually rational and budget-feasible.
	Otherwise, i.e., if lines \ref{line:MSO_initial}-\ref{line:MSO_output} are executed, the proof of monotonicity, and thus of truthfulness and individual rationality, of $\nameref{fig:Main-Submodular-Online}(\rho)$ is virtually identical to the proof of Lemma \ref{lem:SG_monotone}. 
	Similarly, the budget-feasibility is proved exactly like the budget-feasibility of \nameref{fig:Simultaneous Greedy} in Lemma \ref{lem:SG_budget-feasible}.
	
\medskip

\begin{algorithm}[H]
	\DontPrintSemicolon 
	\NoCaptionOfAlgo
	\algotitle{\textnormal{\textsc{GenSm-Online}}}{fig:Main-Submodular-Online.title}
	\WP{With probability $q=0.4$}{Run Dynkin's algorithm and \Return the winner }
	\WP{With probability $1-q$}{  
		$S_1 =  S_2 = T_1 = T_2 = \emptyset$; $B_1 =  B_2 = B$  \label{line:MSO_initial}\; 
		$S=\begin{cases}
		S_j \,, & \text{with probability 1/10, for each } j\in\{1,2\}\\
		T_j \,, & \text{with probability 2/5, for each } j\in\{1,2\}
		\end{cases}$ \label{line:MSO_random_s}\;
		Draw $\xi$ from the binomial distribution $\mathcal B(n,0.5)$ \; 
		Let $A_1$ be the set of the first $\xi$ agents, and $A_2=A\mysetminus A_1$\;
		Reject all the agents in $A_1$ and calculate $x=v(\alg_1(A_1))$ \;
		\For{each $i\in A_2$ as he arrives}{
			Let $\hat{\jmath} \in \argmax_{j\in\{1,2\}} v(i|S_j)$\label{line:MSO_j} \;
			\If{$c_{i}\le \frac{\beta B}{x} v(i|S_{\hat{\jmath}}) \le B_{\hat{\jmath}}$ \label{line:MSO_condition}}{
				$S_{\hat{\jmath}} = S_{\hat{\jmath}} \cup \{i\}$ \;
				$B_{\hat{\jmath}} = B_{\hat{\jmath}} - \frac{\beta B}{x} v(i|S_{\hat{\jmath}})$\;
				With probability $1/2$, $T_{\hat{\jmath}} = T_{\hat{\jmath}} \cup \{i\}$; otherwise, $T_{\hat{\jmath}} = T_{\hat{\jmath}}$ \label{line:MSO_random_t} \;
				Update $S$ \tcc*{{\footnotesize the update is consistent to the choice made in line \ref{line:MSO_random_s}}}
			}
		}
		\Return $S$ \label{line:MSO_output}
	}
	\caption{\textsc{GenSm-Online}$(A, v, \mathbf{c}, B)$} \label{fig:Main-Submodular-Online} 
\end{algorithm}\medskip

Since \nameref{fig:Main-Submodular-Online} is a probability distribution over 
\nameref{fig:Main-Submodular-Online}$(\rho)$ for all possible $\rho$, we conclude that it is universally truthful, individually rational and budget-feasible.
Also, given that Dynkin's algorithm and $\alg_1$ run in polynomial time and that the payments are easily determined, \nameref{fig:Main-Submodular-Online} runs in polynomial time.

	It remains to show that the solution returned by the mechanism is a constant approximation of the offline optimum. It is not hard to see that when the most valuable agent is comparable to the optimal solution, then Dynkin's algorithm suffices to guarantee an overall good performance. In particular, suppose that $\max_{i\in A}v(i) \ge \frac{1}{250} \cdot \opt(A, B)$. Then, with probability $q$ at least $1/e$ of the $1/250$ of the optimal value is returned in expectation (with respect to  the arrival order). Hence, if $X$ is the (possibly empty) set returned by \nameref{fig:Main-Submodular-Online}
	\[\mathbb{E}(X) \ge \frac{q}{e}  \cdot \max_{i\in A}v(i) \ge \frac{q}{e} \cdot \frac{1}{250} \cdot \opt(A, B) \ge \frac{1}{1710} \cdot \opt(A, B)\,. \]
	
	For the case where $\max_{i\in A}v(i) < \frac{1}{250} \cdot \opt(A, B)$, we are going to prove the analog of Lemma \ref{lem:StG_approximation}. 
	First, notice that randomly ordering the elements of $A$ and then picking the first $\xi$, where $\xi$ follows the binomial distribution $\mathcal{B}(n,0.5)$, is equivalent to just picking each element of $A$ with probability $1/2$. This simple observation is crucial, because it allows to still use Lemma \ref{lem:Bei-Leonardi}. So, assume  it is the case that $\opt(A_i,  B) \ge \frac{k-1}{4k} \opt(A, B)$ for $i\in\{1,2\}$, where $k=250$. Unless otherwise stated, all expectations below are conditioned on this fact. Recall that this happens with probability at least $1/2$ as discussed in the beginning of the proof of Lemma \ref{lem:StG_approximation}.
	
	We will follow a similar case analysis as in the proof of Lemma \ref{lem:SG_approximation}, depending on whether the set $R_B$, defined below, is empty or not.  
	Similarly to the notation used in Section \ref{sec:offline}, let $i_1, i_2, \ldots, i_{n-\xi}$ be the agents of $A_2$ ordered according to their arrival. Also, let  $S_{1}^{(k)}$, $B_{1}^{(k)}$, $S_{2}^{(k)}$, $B_{2}^{(k)}$ denote $S_1$, $B_1$, $S_2$, $B_2$, respectively, at the time $i_k$ arrives. We will use $S_{1}$ and $S_{2}$ exclusively for their final versions. Let $R=A_2\mysetminus (S_1\cup S_2)$ be the agents $i_k$ that were rejected, i.e., not added to either $S^{(k)}_1$ or $S^{(k)}_2$. 
	We again partition $R$ depending on why the agents where rejected, i.e., $R_{\mathbf{c}}$ (resp.~$R_B$) contains everyone rejected because the first (resp.~the second) inequality in line \ref{line:MSO_condition} was violated. 
	\smallskip
	
	\noindent \textbf{Case 1.} Assume that  $R_B\neq \emptyset$ and let $i_k\in R_B$. If $j_k$ is the value of $\hat{\jmath}$ chosen in line \ref{line:MSO_j}, then $\frac{\beta  B}{x} v\left( i_k|S^{(k)}_{j_k}\right)  > B^{(k)}_{j_k}$. Using the exact same argument leading to \eqref{eq:SG_approximation_1} (see proof of Lemma \ref{lem:SG_approximation}), we get
	\begin{equation*}
	v\left( S^{(k)}_{j_k}\right) \ge  \frac{x}{\beta}  -  v\left( i_k|S^{(k)}_{j_k}\right) \ge \frac{x}{\beta}  -  \max_{i\in A}v(i)\,. 
	\end{equation*}
	Given the known lower bound on $x$ and upper bound on $\max_{i\in A}v(i)$, this leads to 
	\begin{equation}\label{eq:MSO_approximation_1}
	v(S_{j_k})  \ge \left( \frac{k-1}{4 e k \beta}  -  \frac{1}{k}\right) \cdot \opt(A, B)\,. 
	\end{equation}
	Before we lower bound $\mathbb{E}(v(S))$, it not hard to see that $\mathbb{E}(v(T_j)) \ge \frac{1}{2} v(S_j)$, where the expectation is over the random choices made in line \ref{line:MSO_random_t}. In fact, this is a direct corollary of the non-negativity of $v$ and the following well-known probabilistic property of submodular functions.
	\begin{lemma}[\textnormal{Feige et al.~\shortcite{FeigeMV11}}]
		Let $g :2^X \to \mathbb{R}$ be submodular. Denote by $A[p]$ a random subset of $A$ where each element appears with probability $p$. Then $\mathbb{E}(g(A[p])) (1 - p)\cdot g(\emptyset) + p\cdot g(A)$.
	\end{lemma} 
	By taking the expectation of $v(S)$ over the random choices made in lines \ref{line:MSO_random_t} and \ref{line:MSO_random_s}, we get
	\begin{IEEEeqnarray}{rCl}\label{eq:MSO_approx_exp_1}
		\mathbb{E}(v(S)) & =  & \frac{1}{10}\cdot v(S_1) +\frac{1}{10}\cdot v(S_2) + \frac{2}{5}\cdot \mathbb{E}(v(T_1)) + \frac{2}{5}\cdot \mathbb{E}(v(T_2)) \nonumber\\ 
		& \ge & \frac{1}{10}\cdot v(S_{j_k}) + \frac{2}{5}\cdot \frac{1}{2}\cdot v(S_{j_k}) \nonumber\\
		& \ge  &  \frac{3}{10}\cdot\left( \frac{k-1}{4 e k \beta}  -  \frac{1}{k}\right)\cdot  \opt(A, B) \nonumber\\ 
		& \ge & \frac{1}{512}\cdot  \opt(A, B)  \,.
	\end{IEEEeqnarray}
	\smallskip

	\noindent \textbf{Case 2.} Assume that  $R_B = \emptyset$. Let $C^*$ be an optimal solution for the instance $(A_2, v, \mathbf{c}_{A_2}, B)$ and $C_1= C^* \cap S_1$, $C_2= C^* \cap S_2$, $C_3 = C^* \mysetminus (C_1\cup C_2)$. Recall inequality \eqref{eq:SG_approximation_3} (see proof of Lemma \ref{lem:SG_approximation}):
	\begin{equation}
	\opt(A_2, B) = v(C^*) \le v(C_1) + v(C_2) + v(C_3) \,.  \tag{\ref{eq:SG_approximation_3}}
	\end{equation}
	To upper bound the value of $C_1$ and $C_2$ we need the following result by Feige et al.~\shortcite{FeigeMV11}.
	\begin{theorem}[\textnormal{Feige et al.~\shortcite{FeigeMV11}}]\label{thm:Feige_et_al}
		Let $v:2^A\to \mathbb{R}_{\ge 0}$ be a submodular function and let $T$ denote 
		a random subset of $A$, where each element is sampled independently with probability $1/2$.
		Then $\mathbb{E}(v(T))\ge \frac{1}{4} \opt(A, v, \infty)$.
	\end{theorem}
	By the definition of $T_1, T_2$ and Theorem \ref{thm:Feige_et_al}, we get 
	\begin{equation}\label{eq:MSO_approximation_2}
	v(C_j) \le \opt(S_j,  B) = \opt(S_j, \infty) \le 4 \cdot \mathbb{E}(v(T_j))\,, \text{\ \ for\ \ } j\in\{1,2\} \,. 
	\end{equation}
	For upper bounding $v(C_3)$ recall inequality \eqref{eq:SG_approximation_5} (see proof of Lemma \ref{lem:SG_approximation}):
	\begin{equation}
	v(C_3) \le  v(S_1\cup C_3) + v(S_2\cup C_3)  \,.\tag{\ref{eq:SG_approximation_5}}
	\end{equation}
	Using the same arguments leading to \eqref{eq:SG_approximation_6} (see proof of Lemma \ref{lem:SG_approximation}), we get
	\begin{equation}\label{eq:MSO_approximation_3}
	v(S_j\cup C_3)\le   v(S_j) +  \frac{x}{\beta}\,, \text{\ \ for \ \ } j\in\{1,2\} \,.
	\end{equation}
	We may now combine \eqref{eq:SG_approximation_3}, \eqref{eq:MSO_approximation_2}, \eqref{eq:SG_approximation_5} and \eqref{eq:MSO_approximation_3}. Note that $\mathbb{E}(v(T_j)), j\in\{1,2\}$, below is over the random choices in line \ref{line:MSO_random_t}, while $\mathbb{E}(v(S_j)), j\in\{1,2\}$, is over the random choices in both line \ref{line:MSO_random_t} and line \ref{line:MSO_random_s}.
	\begin{IEEEeqnarray*}{rCl}
		\frac{k-1}{4k} \opt(A, B) & \le &	\opt(A_2, B)\\
		& \le & 4 \cdot \mathbb{E}(v(T_1)) + 4 \cdot \mathbb{E}(v(T_2)) + v(S_1) + v(S_2) +  \frac{2x}{\beta} \\
		& =  &  10 \cdot \mathbb{E}(v(S)) +  \frac{2x}{\beta} \\
		& \le & 10 \cdot \mathbb{E}(v(S)) +  \frac{2}{\beta} \cdot\opt(A, B)\,, 
	\end{IEEEeqnarray*}
	or, equivalently, 
	\begin{IEEEeqnarray}{rCl}\label{eq:MSO_approx_exp_2}
		\mathbb{E}(v(S)) & \ge  & \frac{1}{10} \left(\frac{k-1}{4k}  - \frac{2}{\beta}\right) \cdot \opt(A, B) \nonumber\\ 
		& \ge & \frac{1}{513}\cdot  \opt(A, B)  \,.
	\end{IEEEeqnarray}
	
	Therefore, given that both $A_1$ and $A_2$ contain a good fraction of the optimal budget-feasible solution, the expectation of $v(S)$ is always at least $\frac{1}{513}\cdot  \opt(A, B)$. Coupled with Lemma \ref{lem:Bei-Leonardi}, this means that the unconditional expectation of $v(S)$ is  at least $\frac{1}{2}\cdot \frac{1}{513}\cdot  \opt(A, B)$.
	
	Hence, if $X$ is the set returned by \nameref{fig:Main-Submodular-Online}, by the law of total expectation, we have
	\[\mathbb{E}(X) \ge (1-p)  \cdot \frac{1}{2}\cdot \frac{1}{513}\cdot  \opt(A, B) = \frac{1}{1710} \cdot \opt(A, B)\,. \]
	We conclude that \nameref{fig:Main-Submodular-Online} achieves, in expectation, an $1710$-approximation.
\end{proof}

One immediate consequence of Theorem \ref{thm:Main-Submodular-Online} is the existence of an $O(1)$-approximation algorithm for the non-monotone \emph{Submodular Knapsack Secretary Problem}  (SKS). To the best of our knowledge \nameref{fig:Main-Submodular-Online} is the first such algorithm (see also Remark \ref{rem:SKS}).

Formally, an instance of SKS consists  of a ground  set $A=[n]$,  a non-negative  submodular  objective $v: A \to \mathbb{R}_+$ and  a  given budget $B$. The elements of $A$ arrive in a uniformly random order and each element must be accepted or rejected  immediately upon arrival. 
An algorithm for SKS has access to $n=|A|$, to the costs of items that have arrived (i.e., each cost is revealed upon arrival)
and to a value oracle that, given a subset $S\subseteq A$ of elements that have already arrived, returns $v(S)$.
The objective is to accept a set of elements maximizing $v$ without exceeding the budget.

It is straightforward to see that the only difference of SKS with the online procurement problem studied in this section is the information about the costs.  In SKS there is no notion of misreporting a cost and thus it can be seen as a special case of our online  problem where agents are guaranteed to always reveal their true costs.

\begin{corollary}\label{cor:SKS}
There is an $O(1)$-approximation algorithm for the non-monotone SKS.
\end{corollary}

\begin{remark}\label{rem:SKS}
Bateni et al.\shortcite{BateniHZ13} give an $O(1)$-approximation algorithm for the \emph{monotone} SKS. They claim, without a proof, that their result extends to the non-monotone SKS as well, using the same ideas that work for the Submodular Secretary problem with a cardinality or  a matroid constraint (rather than a knapsack constraint). However, this does not seem to be the case. It is indeed true that these ideas do pair well with cardinality and, more generally, (intersection of) matroid constraints and they have been recently used to obtain fast randomized algorithms with good approximation guarantees \citep{BuchbinderFNS14,FeldmanHK17}.
Unfortunately, to the best of our knowledge, there are no examples in the literature where this \emph{random greedy} approach does work for non-monotone submodular maximization with knapsack constraints.
\end{remark}

\section{Adding Combinatorial Constraints}
\label{sec:combinatorial}
To illustrate the applicability of our approach, we turn to the case where the solution has to satisfy some additional combinatorial constraint. With the exception of additive valuation functions \citep{ABM16,LeonardiMSZ17}, even for \emph{monotone} submodular objectives no polynomial-time 
mechanisms using only value queries 
are known. Here we show that the general approach of \nameref{fig:Main-Submodular} can be utilized to achieve an $O(p)$-approximation for $p$-systems, i.e., for independence systems with \emph{rank quotient} at most $p$. In particular, as stated in Corollary \ref{cor:Constrained-Submodular}, this implies constant factor approximation for cardinality, matroid and matching constraints. As it is shown in Section \ref{sec:lower_bounds}, going beyond independence systems (e.g., require that the solution forms a spanning tree) is hindered by strong impossibility results. 

\begin{definition}
An \emph{independence system} is a pair $(U,\mathcal{I})$, where $U$ is a finite set and $\mathcal{I}\subseteq 2^U$ is a family of subsets, whose members are called the \emph{independent sets} of $U$ and satisfy: 
\begin{enumerate}[label=\emph{(\roman*)},itemsep=4pt,topsep=4pt]
	\item $\emptyset \in \mathcal{I}$, and
	\item if $B\in \mathcal{I}$ and $A\subseteq B$, then $A\in \mathcal{I}$.
\end{enumerate}
Given a set $S\subseteq U$, a maximal independent set contained in $S$ is called a \emph{basis} of $S$. The \emph{upper rank} $\mathrm{ur}(S)$ (resp.~the \emph{lower rank} $\mathrm{lr}(S)$) is defined as the cardinality of a largest (resp.~smallest) basis of $S$. A \emph{$p$-system} $(U,\mathcal{I})$ is an independence system such that $\max_{S\subseteq U}\frac{\mathrm{ur}(S)}{\mathrm{lr}(S)} \le p$.
\end{definition}

For the sake of readability, we present the case of \emph{monotone} submodular objectives here; the non-monotone case is deferred to Appendix \ref{app:sec_combinatorial}. A technical highlight of our analysis, later used for the non-monotone case as well,  is Claim \ref{claim:p-system}. The claim crucially depends on the order we consider the agents, in order to bound the value lost because of the $p$-system constraint.

As usual, we assume the existence of an independence oracle. In particular, when we write that $\mathcal{I}$ is part of the input of the mechanism, we mean that the mechanism has access to a membership oracle for $\mathcal{I}$. 
The parameter $\beta$ is later set to $13/3$. $\alg_3$ in line \ref{line:MC_submod_knapsack_p-system} can be any polynomial time approximation algorithm for monotone submodular maximization subject to a knapsack and a $p$-system constraint. Here we assume the $(p+3)$-approximation algorithm of Badanidiyuru and Vondr{\'{a}}k \shortcite{BadanidiyuruV14}.

\medskip

\begin{algorithm}[H]
	\DontPrintSemicolon 
	\NoCaptionOfAlgo
	\algotitle{\textnormal{\textsc{MonSm-Con\-strained}}}{fig:Main-Constrained.title}
	\WP{With probability $q=0.2$}{\Return $i^*\in \argmax_{i\in A} v(i)$ }
	\WP{With probability $1-q$}{  
		Put each agent of $A$ in either $A_1$ or $A_2$ independently at random with probability $\frac{1}{2}$ \;
		$x=v(\alg_3(A_1))$  \tcc*{{\footnotesize a $(p+3)$-approximation of $\opt(A_1, v, \mathbf{c}_{A_1}, B)$}}  \label{line:MC_submod_knapsack_p-system}  
		$S= \emptyset$; $B_R = B$; $U=A_2$  \;
		\While{$U\neq \emptyset$}{
			Let $\hat{\imath} \in \argmax_{i\in U} v(i|S)$ \label{line:MC_argmax}\;
			\If{$c_{\hat{\imath}}\le \frac{\beta  B}{x} v(\hat{\imath}|S) \le B_R$  {\rm{\textbf {and}}} $S \cup \{\hat{\imath}\} \in \mathcal{I}$ \label{line:MC_condition}}{
				$S = S \cup \{\hat{\imath}\}$ \;
				$B_R = B_R - \frac{\beta  B}{x} v(\hat{\imath}|S)$\;
			}
			$U = U\mysetminus\{\hat{\imath}\}$ \label{line:MC_rejection}\;
		}
		\Return $S$
	} 
	\caption{\textsc{MonSm-Constrained}$(A, \mathcal{I}, v, \mathbf{c}, B)$} \label{fig:Main-Constrained} 
\end{algorithm}

\begin{theorem}\label{thm:Constrained-Submodular}
	Assuming that the solution has to be an independent set of a $p$-system, there is a  universally truthful, individually rational, budget-feasible, $O(p)$-approximation mechanism that runs in polynomial time for (non-monotone) submodular objectives.
\end{theorem}

\begin{proof}
	The proof of the theorem for the non-monotone case is deferred to Appendix \ref{app:sec_combinatorial}. Here we prove that \nameref{fig:Main-Constrained} above has all the stated properties for monotone submodular objectives.
	First, we observe that $S$ starts as an independent set, namely the empty set, and it is expanded only if it remains an independent set. Hence, at the end \nameref{fig:Main-Constrained} does return a feasible solution, i.e.,  $S$ is in $\mathcal{I}$.

At this point, following the same reasoning used for \nameref{fig:Main-Submodular} and \nameref{fig:Main-Submodular-Online}, it should be easy to see that \nameref{fig:Main-Constrained} is universally truthful, individually rational, budget-feasible, and runs in polynomial time.

	
Next we show that the solution returned by the mechanism is an $O(p)$-approximation of the optimum. 
First, suppose that $\max_{i\in A}v(i) \ge  \opt(A, B) / ({26(p + 10)})$. Then, for the set $S$ returned by \nameref{fig:Main-Constrained}, we have
$\mathbb{E}(v(S)) \ge q\cdot \max_{i\in A}v(i) \ge \frac{1}{5} \cdot \opt(A, B) / ({26(p + 10)}) \ge \opt(A, B)/({138(p+10)})$.


For the case where $\max_{i\in A}v(i) <  \opt(A, B)/({26(p + 10)})$, we follow the same notation and the same high level approach as with the approximation guarantees of \nameref{fig:Main-Submodular} and \nameref{fig:Main-Submodular-Online}. So, 
$i_1, i_2, \ldots, i_{|A_2|}$ are the agents of $A_2$ in the order considered by the mechanism. By $S^{(k)}$ and $B_{R}^{(k)}$ we denote $S$ and $B_R$, respectively, at the time $i_k$ arrives, and we only use $S$ for the final set returned. The set $R=A_2\mysetminus S$ contains the agents $i_k$ that were not added to $S^{(k)}$ and it is further partitioned to 
\[
R_{\mathbf{c}} = \bigg\lbrace i_k \,\Big|\, \frac{\beta  B}{x} v\Big( i_k\,|\,S^{(k)}\Big)  < c_{\hat{\imath}}\bigg\rbrace, \quad
R_B  = \bigg\lbrace i_k \,\Big|\, B_R^{(k)} < \frac{\beta  B}{x} v\Big( i_k\,|\,S^{(k)}\Big) \bigg\rbrace 
\quad\text{and}\quad 
R_{\mathcal{I}}  = R\mysetminus (R_{\mathbf{c}}\cup R_B) \,.
\]

Assume that $\opt(A_i,  B) \ge \frac{k-1}{4k} \opt(A, B)$ for $i\in\{1,2\}$, where $k=26(p + 10)$. Thus, $x=v(\alg_1(A_1))\ge \frac{k-1}{4(p+3)k} \opt(A, B)$. Recall that this does happen with probability at least $\frac12$, as discussed in the beginning of the proof of Lemma \ref{lem:StG_approximation}.
\smallskip

\noindent \textbf{Case 1.} Assume that  $R_B\neq \emptyset$. Let $i_k\in R_B$, i.e., $\frac{\beta  B}{x} v\left( i_k|S^{(k)}\right)  > B^{(k)}_{R}$. Using the same argument as in the proof of Lemma \ref{lem:SG_approximation}, we get
$v\big( S^{(k)}\big) \ge  \frac{x}{\beta}  -  \max_{i\in A}v(i)$
and, given the known bounds on $x$ and $\max_{i\in A}v(i)$, this leads to $v(S)  \ge \big( \frac{k-1}{4(p+3)k \beta} -  \frac{1}{k}\big) \cdot \opt(A, B)$.

By substituting  $k=26(p + 10)$ and $\beta=\frac{13}{3}$, it is a matter of simple calculations to get
\begin{equation}\label{eq:MC_approximation_1}
v(S)  \ge \frac{5}{276(p+10)} \cdot \opt(A, B)\,. 
\end{equation}
\smallskip

\noindent \textbf{Case 2.} Assume that  $R_B = \emptyset$ and let $C^*$ be an optimal solution for the instance $(A_2, v, \mathbf{c}_{A_2}, B)$. By monotonicity, we have
\begin{equation}\label{eq:MC_opt_ub}
\opt(A_2, B) = v(C^*) \le v(S\cup C^*) \,.
\end{equation}
Because of the $p$-system constraint, however, deriving the analog of inequality \eqref{eq:SG_approximation_6} needs some extra work. 
By Definition \ref{def:SM}\ref{def:sub3}, we have
\begin{IEEEeqnarray}{rCl}
	v(S\cup C^*)& \le  & v(S) + \sum_{i_k\in C^*\mysetminus S} v( i_{k} |S) 
	 \le  v(S) + \sum_{i_k\in C^*\cap R_{\mathbf{c}}} v( i_{k} |S) + \sum_{i_k\in C^*\cap R_{\mathcal{I}}} v( i_{k} |S) \,. \label{eq:MC_approximation_2}
\end{IEEEeqnarray}
We may upper bound the first sum using the fact that all agents involved got rejected because they had very low marginal value per cost ratio. That is,
\begin{equation}\label{eq:MC_approximation_3}
 \sum_{i_k\in C^*\cap R_{\mathbf{c}}} v( i_{k} |S) \le    \sum_{i_k\in C^*\cap R_{\mathbf{c}}} v\left( i_{k} |S^{(k)}\right) <  \frac{x}{\beta  B} \sum_{i_k\in C^*\cap R_{\mathbf{c}}}  c_{i_k} \le \frac{x}{\beta} \le \frac{\opt(A, B)}{\beta}\,.
\end{equation}
For the second sum we prove the following result that crucially relies on the fact that agents are examined in decreasing marginal value.
\begin{claim}\label{claim:p-system}
$\sum_{i_k\in C^*\cap R_{\mathcal{I}}} v( i_{k} |S) \le p \cdot v(S)\,.$
\end{claim}

\begin{proof}[Proof of Claim \ref{claim:p-system}] \renewcommand{\qedsymbol}{{\large $\triangleleft$}}
Recall that when we index agents we  follow the ordering imposed by the mechanism, i.e., $i_k$ is always the agent picked at the $k$th execution of line \ref{line:MC_argmax} of \nameref{fig:Main-Constrained}.

Suppose that there is a mapping $f:C^*\cap R_{\mathcal{I}}\to S$ such that
\begin{enumerate}[label=\emph{(\roman*)}]
	\item if $f(i_k)=i_{\ell}$, then $v\left(  i_k | S^{(k)}\right) \le v\left( i_{\ell}| S^{(\ell)}\right) $ for all $i_{k}\in C^*\cap R_{\mathcal{I}}$, and
	\item $\left| f^{-1}(i_{\ell})\right|\le p$ for all $i_{\ell}\in S$.
\end{enumerate}
We slightly abuse the notation and write $S^{f(i_{k})}$ instead of $S^{(\ell)}$ when $f(i_k)=i_{\ell}$. The existence of $f$ implies that
\begin{IEEEeqnarray*}{rCl}
\sum_{i_k\in C^*\cap R_{\mathcal{I}}} \!v( i_{k} | S) &\le& \!\sum_{i_k\in C^*\cap R_{\mathcal{I}}} \!v\left(  i_{k} | S^{(k)}\right) 
 \le  \!\sum_{i_k\in C^*\cap R_{\mathcal{I}}} \!v\left(  f(i_{k}) \,|\, S^{f(i_{k})}\right) 
\le   p\cdot  \!\sum_{i_{\ell}\in S} v\left(  i_{\ell} | S^{(\ell)}\right) =  p\cdot  v(S) \,.
\end{IEEEeqnarray*}
The first inequality follows from the submodularity of $v$, while the second and third inequalities follow from \emph{(i)} and \emph{(ii)}, respectively.

Next, we are going to construct such an $f$. Let $S= \{i_{a_1}, i_{a_2}, \allowbreak \ldots, \allowbreak i_{a_{s}}\}$ and $C^* \cap R_{\mathcal{I}}= \{i_{b_1}, i_{b_2}, \allowbreak \ldots, \allowbreak i_{b_{t}}\}$, where both $(a_i)_{i=1}^{s}$ and $(b_i)_{i=1}^{t}$ are  subsequences of $1, 2, \ldots, |A_2|$. We are going to map the first $p$ elements of $C^*\cap R_{\mathcal{I}}$, $i_{b_1},\ldots, i_{b_p}$, to $i_{a_1}$, the next $p$ elements $i_{b_{p+1}},\ldots, i_{b_{2p}}$, to $i_{a_2}$, and so on. That is,
$f(i_{b_{j}}) = i_{a_{\left\lceil j/p \right\rceil}}$. 

It is straightforward that $f$ satisfies property \emph{(ii)}. In order to prove property  \emph{(i)}, it suffices to show that for all $j\in \{1,2,\ldots, t\}$, agent $i_{b_{j}}$ is considered by \nameref{fig:Main-Constrained} after agent $f(i_{b_{j}})$. Indeed, if that was the case, by the definition of $\hat{\imath}$ in line \ref{line:MC_argmax} and submodularity, we would get
\[ v\left( i_{a_{\left\lceil j/p \right\rceil}} \,|\, S^{(a_{\left\lceil j/p \right\rceil})}\right) \ge v\left( i_{b_{j}} \,|\, S^{(a_{\left\lceil j/p \right\rceil})}\right) \ge  v\left( i_{b_{j}} \,|\, S^{(b_{j})}\right) \,, \]
for all $i_{b_{j}}\in C^*\cap R_{\mathcal{I}}$, as desired. Suppose, towards a contradiction, that there is some $k\in \{1,2,\ldots, t\}$, such that $b_{k}<a_{\left\lceil k/p \right\rceil}$; in fact, suppose $k$ is the smallest such index. Consider the sets $T=\{i_{a_1}, i_{a_2}, \ldots,i_{a_{\left\lceil k/p \right\rceil - 1}}\}\subseteq S$ and $Q = \{i_{b_1}, i_{b_2}, \ldots,i_{b_{k}}\}\subseteq C^*\cap R_{\mathcal{I}}$. 
By construction, $T\in\mathcal{I}$. Moreover, we claim that $T$ is maximally independent in $T\cup Q$. Indeed, each $i_{b_{\tau}}\in Q$ was rejected because $S^{(b_{\tau})}\cup\{i_{b_{\tau}}\}\notin \mathcal{I}$, and since $S^{(b_{\tau})} \subseteq T$ we get $T\cup\{i_{b_{\tau}}\}\notin \mathcal{I}$.
This implies that $\mathrm{lr}(T\cup Q)\le |T|$. On the other hand, $Q\in\mathcal{I}$ because $Q\subseteq C^*\in\mathcal{I}$. As a result $\mathrm{ur}(T\cup Q)\ge |Q|$.
However, notice that 
\[p\cdot |T| = p \left( \left\lceil k/p \right\rceil -1\right)  <   p \left(  k/p +1 -1\right) = k = |Q|\,. \]
Thus, $\frac{\mathrm{ur}(T\cup Q)}{\mathrm{lr}(T\cup Q)} \ge \frac{|Q|}{|T|} >p$, contradicting the fact that $(A,\mathcal{I})$ is a $p$-system.
We conclude that $f$ satisfies both \emph{(i)} and \emph{(ii)}, and therefore, $\sum_{i_k\in C^*\cap R_{\mathcal{I}}} v( i_{k} |S) \le p \cdot v(S)$.
\end{proof}
Now, combining \eqref{eq:MC_opt_ub}, \eqref{eq:MC_approximation_2}, \eqref{eq:MC_approximation_3}, and Claim \ref{claim:p-system}, we have
$\opt(A_2, B) \le (p+1)\cdot v(S) +  \frac{\opt(A, B)}{\beta}$,
and using the lower bound on $\opt(A_2, B)$,
$v(S) \ge \frac{1}{p+1} \cdot  \big( \frac{k-1}{4k} - \frac{1}{\beta}\big)  \opt(A, B)$.
Again, by substituting  $k$ and $\beta$, it is a matter of calculations to get
\begin{equation}\label{eq:MC_approximation_4}
v(S)  \ge \frac{5}{276(p+10)} \cdot \opt(A, B)\,. 
\end{equation}

By Lemma \ref{lem:Bei-Leonardi},  both \eqref{eq:MC_approximation_1} and \eqref{eq:MC_approximation_4} hold with probability at least $1/2$. 
Hence,
\[\mathbb{E}(v(S)) \ge (1-q)\cdot\frac{1}{2}\cdot\frac{5}{276(p+10)} \cdot \opt(A, B)  = \frac{1}{138(p+10)} \cdot \opt(A, B) \,. \qedhere\]
\end{proof}

For matroid constraints we have $p=1$ and for matching constraints $p=2$. Since cardinality constraints are a special case of matroid constraint, we directly get the following.

\begin{corollary}\label{cor:Constrained-Submodular}
For cardinality, matroid and matching constraints, there is a  universally truthful, budget-feasible $O(1)$-approximation mechanism for (non-monotone) submodular objectives.
\end{corollary}

\section{Lower Bounds}
\label{sec:lower_bounds}
In the value query model there is a strong lower bound on the number of queries for deterministic algorithms for monotone XOS objectives due to Singer \shortcite{Singer10}.  This result is based on a lower bound of Mirrokni et al.~\citep{MirrokniSV08} on welfare maximization in combinatorial auctions. As the latter also holds for randomized algorithms, so does Singer's result as well, essentially with the same proof. We restate it here for completeness. Note that it holds even when the costs are public knowledge.

\begin{theorem}[\textnormal{Singer \shortcite{Singer10}}]\label{thm:xos_Singer}
	For any fixed $\varepsilon > 0$, any (randomized) $n^{\frac{1}{2}-\varepsilon}$-appro\-xi\-mation algorithm for monotone XOS function maximization subject to a budget constraint requires exponentially many value queries (in expectation). 
\end{theorem}

When one moves to non-monotone objectives, as it is the case in this work, it is possible to prove even stronger lower bounds. Below we show that for general XOS objectives, exponentially many value queries are needed  for \emph{any} non-trivial approximation even without the budget constraint. As this result applies to the purely  algorithmic setting, it is of independent interest. 

It is known that in many settings there is a  separation between the power of value and demand queries of polynomial size, see, e.g., \citep{BlumrosenN09}. To stress this difference in our setting, recall that in the demand query model, the class of XOS objectives admits a truthful $O(1)$-approximation mechanism with a polynomial number of queries. 

\begin{theorem}\label{thm:xos_nonmonotone}
	For any fixed $\varepsilon > 0$, any (randomized) $n^{1-\varepsilon}$-approximation algorithm for XOS function maximization requires exponentially many value queries (in expectation). 
\end{theorem}

\begin{proof}
	We follow, on a high level, the approach in \citep{MirrokniSV08}.
	Recall that $A=[n]$ and choose a set $R$ of size $|R| = \rho=n/4$ uniformly at random amongst all the subsets of $A$ of size $\rho$. We are going to construct two XOS functions, $v_1$ and $v_2$, that are hard to tell apart, i.e., to distinguish between them with constant probability, an exponential number of value queries will be required. 
	
	For any $T\subseteq A$, let $\alpha_T$ be the additive function that assigns the value $1$ to each $i\in T$ and the value $0$ to each $i\notin T$. For $\tau = n^{\varepsilon/2}/4$, we define $v_1$ as the maximum over all such additive functions on sets of size $\tau$:
	\[v_1(S)=\max_{T\subseteq A : |T| = \tau} \alpha_T(S), \text{ \ for all } S\subseteq A \,. \]
	Further, let $\beta$ be the additive function that assigns the value $1$ to each $i\in R$ and the value $-\rho$ to each $i\notin R$. We define $v_2$ as the maximum between $v_1$ and $\beta$:
	\[v_2(S)=\max\left\lbrace v_1(S), \beta(S) \right\rbrace, \text{ \ for all } S\subseteq A \,. \]
	Clearly, both $v_1$ and $v_2$ are XOS functions since each of them is defined as the maximum of a finite number of additive functions. Also notice that for any $S\nsubseteq R$ we have $v_2(S) = v_1(S)$.
	However, $\opt(A, v_1, \infty) = \tau$ and $\opt(A, v_2, \infty) = \rho = n^{1-\varepsilon/2}\cdot \tau > n^{1-\varepsilon} \cdot \tau$. Hence, any (possibly randomized) algorithm that achieves an approximation ratio smaller or equal to $n^{1-\varepsilon}$ can distinguish between the two functions.
	
	Consider a value query for some set $S$. This query can distinguish between $v_1$ and $v_2$ if and only if $S\subseteq R$ and $|S| > \tau$, and otherwise it will reveal no information about $R$. We will call such an $S$ a \emph{distinguishing set}. For a given $S$ with $|S| > \tau$, the probability that $S\subseteq R$, over the random choice of $R$, is 
	\begin{equation}\label{eq:lb_probability_1}
	\frac{{\rho \choose |S|}}{{n \choose |S|}}\le \frac{\left( \frac{e\rho}{|S|}\right)^{|S|}}{\left( \frac{n}{|S|}\right)^{|S|}} =   \left( \frac{e}{4}\right)^{|S|} < \left( \frac{e}{4}\right)^{\frac{n^{\varepsilon/2}}{4}} ,
	\end{equation}
	using the well-known fact that for $1 \leq k \leq m$ we have
	$\left(\frac{m}{k}\right)^k \leq \binom{m}{k} \leq \left(\frac{em}{k}\right)^k.$
	
	Now, let $q(\cdot)$ be a polynomial and $p\in\left( 0,1\right]$ be a constant. Suppose first that there is a deterministic algorithm that asks queries $S_1, S_2, \ldots, S_{q(n)}$  and distinguishes between $v_1$ and $v_2$ with  probability at least  $p$. Note that the choice for $S_{j}$ can depend on all previous queries $S_1,\dots,S_{j-1}$ as well as the answers of the value query oracle obtained for those sets. Also, the choices made by the algorithm are the same for all non-distinguishing queries regardless of whether we present $v_1$ or $v_2$ to the algorithm. 
	Using a union bound, it then follows that the probability that we distinguish between $v_1$ and $v_2$ is at most
	\begin{equation*}
	\sum_{i=1}^{q(n)}\frac{{\rho \choose |S_i|}}{{n \choose |S_i|}} < q(n)\left( \frac{e}{4}\right)^{\frac{n^{\varepsilon/2}}{4}} = o(1) .
	\end{equation*}
	which contradicts $p$ being constant. In case of a randomized algorithm, we can condition on the random bits of the algorithm. Averaging over the choices of the random bits, we are still only able to distinguish between $v_1$ and $v_2$ with exponentially small probability.
\end{proof}

One immediate consequence of Theorem \ref{thm:xos_nonmonotone} is that when we care for constant approximation ratios, the result of Theorem \ref{thm:Main-Submodular} is (asymptotically) the \emph{best possible} for budget-feasible mechanism design. General submodular objectives is the broadest class of well studied non-monotone functions one could hope for, even for randomized mechanisms.
\medskip

\noindent {\bf Combinatorial Constraints.}
We now turn to the problem of maximizing subject to additional constraints on top of the budget constraint.
To further motivate our restriction to $p$-system constraints, we restate here a lower bound of Badanidiyuru and Vondr{\'{a}}k \shortcite{BadanidiyuruV14}: for independence system constraints one cannot achieve an approximation factor better than $\max_{S\subseteq U}\frac{\mathrm{ur}(S)}{\mathrm{lr}(S)}$ with a polynomial number of queries. Thus, the result of Theorem \ref{thm:Constrained-Submodular} is asymptotically optimal.

\begin{theorem}[\textnormal{Badanidiyuru and Vondr{\'{a}}k \shortcite{BadanidiyuruV14}}]\label{thm:p-systems_BV}
	For any fixed $\varepsilon > 0$, any (randomized) $(p+\varepsilon)$-appro\-xi\-mation algorithm for additive function maximization 
	subject to $p$-system constraints requires exponentially many independence oracle queries (in expectation).
\end{theorem}

As we mentioned in the beginning of Section \ref{sec:combinatorial}, we cannot really go beyond independence systems and have any non-trivial approximation guarantee in polynomial time. This is illustrated in Theorem \ref{thm:combinatorial_lb_1} and Corollary \ref{cor:combinatorial_lb_2} below.  
Theorem \ref{thm:combinatorial_lb_1} generalizes Singer's \shortcite{Singer10} strong impossibility result for deterministically ``hiring a team of agents'' to \emph{any} constraint that is not downward closed below. Note that it holds even for super-constant approximation ratios, even for the special case of additive objectives, irrespectively of any complexity assumptions. 

\begin{theorem}\label{thm:combinatorial_lb_1}
	Let $\mathcal{F}\subseteq 2^A$ be any collection of feasible sets that is not downward closed. Then there is no deterministic, truthful, individually rational, budget-feasible mechanism achieving a bounded approximation when restricted on $\mathcal{F}$, even for additive objectives.
\end{theorem}

\begin{proof}
	Since $\mathcal{F}$ is not downward closed, there is some $F\in \mathcal{F}$ with $|F|\ge 2$ which is minimally feasible, i.e., if $S\subseteq F$ and $S\in \mathcal{F}$, then $S=F$. 
	
	Towards a contradiction, suppose that there is a deterministic, truthful, budget-feasible, $\alpha$-appro\-xi\-ma\-tion mechanism $\alg$ for additive objectives, where  $\alpha = \alpha(n) >1$. Consider the following instance on $A$ where $v$ is additive: for each agent $i\in F$, $v(i) = 1/|F|, c_i = \varepsilon \ll B/|F|$, while for each agent $i \in A\mysetminus F$, $v(i) = \delta < 1/\alpha, c_i = B$. All the $\mathcal{F}$-feasible and budget-feasible solutions are $F$ and, possibly, some of the singletons outside of $F$. 
	If $\alg$  returns any solution other than $F$, then 
	$v(\alg(A, v, \mathbf{c}, B)) \le \delta  <\frac{1}{\alpha} = \frac{1}{\alpha} \cdot \opt(A, B)$,	which contradicts the approximation guarantee of $\alg$. So, $\alg$ should return $F$. 
	
	However, the latter is true even if we slightly modify the instance, so that for a specific $j\in F$, $c_j = B - (|F|-1)\cdot\varepsilon$. Therefore, in the original instance, the threshold payment for $j$ is at least $B - (|F|-1)\cdot\varepsilon$. In fact, due to symmetry, all the threshold payments in the original instance should be  at least $B - (|F|-1)\cdot\varepsilon$. Since $|F|\ge 2$ and $B - (|F|-1)\cdot\varepsilon \approx B$, this contradicts the budget-feasibility of $\alg$.
\end{proof}

The next corollary of Theorem \ref{thm:p-systems_BV} states that under general combinatorial constraints it is not possible to achieve any non-trivial approximation with polynomially many queries. While it is not hard to prove it directly, given Theorem \ref{thm:p-systems_BV} it suffices to notice that such a lower bound holds even for general independence systems. Indeed, there are cases where $\frac{\mathrm{ur}(U)}{\mathrm{lr}(U)}$ is $\Theta(n)$ like the $(n-1)$-systems of independent sets of star graphs.

\begin{corollary}\label{cor:combinatorial_lb_2}
	For any fixed $\varepsilon > 0$, any (randomized) $n^{1-\varepsilon}$-appro\-xi\-mation algorithm for additive function maximization 
	subject to general feasibility constraints requires exponentially many queries (in expectation). 
\end{corollary}

\section{Discussion}
\label{sec:disc}
We already discussed in the \nameref{sec:intro}  that designing \emph{deterministic} budget-feasible mechanisms has been elusive. Positive results are only known for specific well-behaved objectives \citep{Singer10,ChenGL11,Singer12,ABM16,HorelIM14,DobzinskiPS11,ABM17} and, even worse, beyond monotone submodular  valuation functions no deterministic $O(1)$-approximation mechanism is known, irrespectively of time or query complexity.
We consider obtaining deterministic, budget-feasible, $O(1)$-approximation  mechanisms---or showing that they do not exist---the most intriguing related open problem.

While our results provide a proof of concept with respect to what is asymptotically possible with polynomial-time, truthful mechanisms, the constants involved are very far from being practical. Although we do not claim that the different parameters appearing in the description and the analysis of our mechanisms are optimized, they had to be carefully chosen and we suspect there is not much room for improvement. Bringing down these approximation factors is another interesting direction.

Finally, it is mentioned in Remark \ref{rem:7-approx} that the high level approach of \nameref{fig:Simultaneous Greedy} can be turned into a deterministic 7-approximation  algorithm. We believe that it is worth exploring other possible applications of the high level approach of \nameref{fig:Simultaneous Greedy}, both in mechanism design and in constrained non-monotone submodular maximization.

\newpage

\appendix

\section{Proof of Theorem \ref{thm:Constrained-Submodular} for the Non-Monotone Case}
\label{app:sec_combinatorial}

\begin{rtheorem}{Theorem}{\ref{thm:Constrained-Submodular}}
	Assuming that the solution has to be an independent set of a $p$-system, there is a  universally truthful, individually rational, budget-feasible, $O(p)$-approximation mechanism that runs in polynomial time for non-monotone submodular objectives.
\end{rtheorem}

\begin{proof}

Here we move on to the case of \emph{non-monotone} submodular objectives. \nameref{fig:nMain-Constrained} is a modification of \nameref{fig:Main-Submodular}  that maintains a set $F$ of ``feasible'' pairs, i.e., of pairs $(i,j)$ such that  $S_j \cup \{i\}$ is an independent set. In each step, the best such pair $(\hat{\imath},\hat{\jmath})$ is chosen and, given that $v(\hat{\imath}|S_{\hat{\jmath}})$ is neither too high nor too low, $\hat{\imath}$  is added to $S_{\hat{\jmath}}$.
The parameter $\beta$ is  $8.5$ and $\alg_4$ in line \ref{line:nMC_submod_knapsack_p-system} can be any polynomial time approximation algorithm for non-monotone submodular maximization subject to a knapsack and a $p$-system constraint. Here we assume the $\frac{(1+\varepsilon)(p+1)(2p+3)}{p}$-approximation algorithm of Mirzasoleiman et al.~\shortcite{MirzasoleimanBK16} for $\varepsilon=10^{-3}$.
\medskip

\begin{algorithm}[H]
	\DontPrintSemicolon 
	\NoCaptionOfAlgo
	\algotitle{\textnormal{\textsc{GenSm-Con\-strained}}}{fig:nMain-Constrained.title}
	\WP{With probability $q=1/3$}{\Return $i^*\in \argmax_{i\in A} v(i)$ }
	\WP{With probability $1-q$}{  
		Put each agent of $A$ in either $A_1$ or $A_2$ independently at random with probability $\frac{1}{2}$ \;
		$x=v(\alg_4(A_1))$  \tcc*{{\footnotesize a $(1+\varepsilon)(p+1)(2p+3)/p$-approximation of $\opt(A_1, v, \mathbf{c}_{A_1}, B)$}}  \label{line:nMC_submod_knapsack_p-system}  
		$S_1 =  S_2 = \emptyset$; $B_1 =  B_2 = B$; $U=A_2$  \;
		$F=\{(i,j)\,|\, i\in U, j\in\{1,2\} \text{\ and\ } S_j \cup \{i\} \in \mathcal{I}   \}$  \tcc*{{\footnotesize all ``feasible'' pairs}}  
		\While{$F\neq \emptyset$}{
			Let $(\hat{\imath},\hat{\jmath}) \in \argmax_{(i,j)\in F} v(i|S_j)$ \label{line:nMC_argmax}\;
			\If{$c_{\hat{\imath}}\le \frac{\beta  B}{x} v(\hat{\imath}|S_{\hat{\jmath}}) \le B_{\hat{\jmath}}$ \label{line:nMC_condition}}{
			$S_{\hat{\jmath}} = S_{\hat{\jmath}} \cup \{\hat{\imath}\}$ \;
			$B_{\hat{\jmath}} = B_{\hat{\jmath}} - \frac{\beta  B}{x} v(\hat{\imath}|S_{\hat{\jmath}})$\;}
		$U = U\mysetminus\{\hat{\imath}\}$ \label{line:nMC_rejection}\;
		Update $F$ \;
		}
		\For{$j\in \{1,2\}$}{
			$T_j=\alg_2(S_j)$ \tcc*{{\footnotesize a 2-approximate solution with respect to  $\opt(S_j, v, \mathbf{c}_{S_j}, \infty)$}} \label{line:nMC_unconstrained}
		}
		Let $S$ be the best solution among $S_1, S_2, T_1, T_2$ \label{line:nMC_max_soln}\;
		\Return $S$
	} 
	\caption{\textsc{GenSm-Constrained}$(A, \mathcal{I}, v, \mathbf{c}, B)$} \label{fig:nMain-Constrained} 
\end{algorithm}\medskip

Clearly, $S_1, S_2$ start as independent sets and they are expanded only if they remain independent sets. As subsets of independent sets, $T_1, T_2$ are independent sets as well. Hence, \nameref{fig:nMain-Constrained} does return a solution $S\in\mathcal{I}$.
	
Like in the monotone case, following the reasoning used for \nameref{fig:Main-Submodular} and \nameref{fig:Main-Submodular-Online}, it is easy to prove universal truthfulness, individual rationality, budget-feasibility, and---given polynomial-time oracles---polynomial running time.
What is left to show is that $\mathbb{E}(v(S))$ is an $O(p)$-approximation of $\opt(A, B)$. 

First, suppose that $\max_{i\in A}v(i) \ge  \frac{1}{136(p + 6)}  \cdot \opt(A, B)$. Then, for the set $S$ returned by \nameref{fig:nMain-Constrained},
\[\mathbb{E}(v(S)) \ge q\cdot \max_{i\in A}v(i) \ge \frac{1}{3} \cdot \frac{1}{136(p + 6)}\cdot \opt(A, B) > \frac{1}{410(p+6)} \cdot \opt(A, B) \,. \]

When $\max_{i\in A}v(i) <  \frac{1}{136(p + 6)}  \cdot \opt(A, B)$, we may follow the same approach as with the our other proofs. Recall the notation. That is, $i_1, i_2, \ldots, i_{|A_2|}$ are the agents of $A_2$ in the order considered by the mechanism and $j_1, \ldots, j_{|A_2|}$ are the corresponding $\hat{\jmath}$ selected in the $k$th execution of line \ref{line:nMC_argmax}.\footnote{In case not all agents are considered, what remains in $F$ is arbitrarily indexed and paired with some $\hat{\jmath}$. This is as if we had a few dummy iterations at the end of the \textbf{while} loop in order to exhaust all agents by rejecting them one by one.} By $S_{j}^{(k)}$ and $B_{j}^{(k)}$ we denote $S_j$ and $B_j$, respectively, at the time $i_k$ is selected. We only use $S_1, S_2, B_1, B_2$ for the final version of the corresponding set or quantity.
The set $R=A_2\mysetminus (S_1 \cup S_2)$ contains the agents $i_k$ that were not added to $S_{j_k}^{(k)}$ and it is further partitioned to 
$R_{\mathbf{c}} = \big\lbrace i_k \,|\, \frac{\beta  B}{x} v\big( i_k|S_{i_k}^{(k)}\big)  < c_{\hat{\imath}}\big\rbrace $, $R_B  = \big\lbrace i_k \,|\, B_{i_k}^{(k)} < \frac{\beta  B}{x} v\big( i_k|S_{i_k}^{(k)}\big) \big\rbrace$,   and $R_{\mathcal{I}}  = R\mysetminus (R_{\mathbf{c}}\cup R_B)$.

Recall that Lemma \ref{lem:Bei-Leonardi} guarantees that $\opt(A_i,  B) \ge \frac{k-1}{4k} \opt(A, B)$ for $i\in\{1,2\}$, where $k=136(p + 6)$, happens with probability at least $1/2$. Assume this is indeed the case. Therefore, $x=v(\alg_1(A_1))\ge \frac{(k-1)p}{4k(1+\varepsilon)(p+1)(2p+3)} \opt(A, B)$.
\smallskip

\noindent \textbf{Case 1.} Assume that  $R_B\neq \emptyset$. By repeating the analysis of Case 1 in the proof of Lemma \ref{lem:SG_approximation}, we get
\begin{equation*}
v(S)  \ge \left( \frac{(k-1)p}{4k(1+\varepsilon)(p+1)(2p+3)\beta} -  \frac{1}{k}\right) \cdot \opt(A, B)\,. 
\end{equation*}
By substituting  $k=136(p + 6)$, $\beta=8.5$ and $\varepsilon=10^{-3}$, it is a matter of simple calculations to get
\begin{equation}\label{eq:nMC_approximation_1}
v(S)  \ge \frac{1}{136(p + 6)} \cdot \opt(A, B)\,. 
\end{equation}
\smallskip

\noindent \textbf{Case 2.} Next, assume that  $R_B = \emptyset$. Let $C^*$ be an optimal solution for the instance $(A_2, v, \mathbf{c}_{A_2}, B)$ and $C_1= C^* \cap S_1$, $C_2= C^* \cap S_2$, $C_3 = C^* \mysetminus (C_1\cup C_2)$. By subadditivity (recall inequality \eqref{eq:SG_approximation_3}) and the fact that $T_j=\alg_2(S_j), j\in\{1,2\}$, is a 2-approximate solution with respect to  $\opt(S_j, \infty)$ we get
\begin{equation}\label{eq:nMC_approximation_3b}
\opt(A_2, B) = v(C^*) \le v(C_1) + v(C_2) + v(C_3) \le 2v(T_1) + 2v(T_2) + v(C_3)\,. 
\end{equation}

For  $v(C_3)$ recall inequality \eqref{eq:SG_approximation_5} (see proof of Lemma \ref{lem:SG_approximation}):
\begin{equation}
v(C_3) \le  v(S_1\cup C_3) + v(S_2\cup C_3)  \,.\tag{\ref{eq:SG_approximation_5}}
\end{equation}
To upper bound $v(S_j\cup C_3)$ we work like in the proof of \ref{thm:Constrained-Submodular} because of the $p$-system constraint.
By Definition \ref{def:SM}\ref{def:sub3}, we have
	\begin{IEEEeqnarray}{rCl}
		v(S_j\cup C_3)& \le  & v(S_j) + \sum_{i_k\in C_3} v( i_{k} |S_j) \nonumber \\ 
		& \le & v(S_j) + \sum_{i_k\in C_3\cap R_{\mathbf{c}}} v( i_{k} |S_j) + \sum_{i_k\in C_3\cap R_{\mathcal{I}}} v( i_{k} |S_j) \,. \label{eq:nMC_approximation_2}
	\end{IEEEeqnarray}
	We upper bound the first sum exactly as in \eqref{eq:MC_approximation_3}:
	\begin{equation}\label{eq:nMC_approximation_3}
	\sum_{i_k\in C_3\cap R_{\mathbf{c}}} v( i_{k} |S_j) \le    \sum_{i_k\in C_3\cap R_{\mathbf{c}}} v\left( i_{k} |S_j^{(k)}\right) <  \frac{x}{\beta  B} \sum_{i_k\in C_3\cap R_{\mathbf{c}}}  c_{i_k} \le \frac{x}{\beta} \le \frac{\opt(A, B)}{\beta}\,.
	\end{equation}
	For the second sum we have the analog of Claim \ref{claim:p-system}. Recall that we never used the monotonicity of $v$ in the proof of Claim \ref{claim:p-system}. With just minor changes in notation, we can prove the following.
	\begin{claim}\label{claim:nonmonotone_p-system}
		For both $j\in\{1,2\}$, $\sum_{i_k\in C_3\cap R_{\mathcal{I}}} v( i_{k} |S_j) \le p \cdot v(S_j)\,.$
	\end{claim}
Now, combining \eqref{eq:nMC_approximation_3b}, \eqref{eq:SG_approximation_5}, \eqref{eq:nMC_approximation_2}, \eqref{eq:nMC_approximation_3}, and Claim \ref{claim:nonmonotone_p-system}, we have
\begin{equation*}
\opt(A_2, B) \le 2 v(T_1) + 2 v(T_2) + (p+1) v(S_1) + (p+1) v(S_2) +  2\frac{\opt(A, B)}{\beta}\,,
\end{equation*}
and, using the definition of $S$ and the lower bound on $\opt(A_2, B)$,
\begin{equation*}
v(S) \ge \frac{1}{2p+6} \cdot  \left( \frac{k-1}{4k} - \frac{2}{\beta}\right)  \opt(A, B)\,.
\end{equation*}
By substituting  $k$ and $\beta$, it is a matter of calculations to get
\begin{equation}\label{eq:nMC_approximation_4}
v(S)  \ge \frac{1}{136(p + 6)} \cdot \opt(A, B)\,. 
\end{equation}
Since, due to Lemma \ref{lem:Bei-Leonardi}, both \eqref{eq:nMC_approximation_1} and \eqref{eq:nMC_approximation_4} hold with probability at least $1/2$, we have
	\[\mathbb{E}(v(S)) \ge (1-q)\cdot\frac{1}{2}\cdot\frac{1}{136(p + 6)} \cdot \opt(A, B) > \frac{1}{410(p+10)} \cdot \opt(A, B) \,, \]
thus concluding the proof.
\end{proof}

\newpage

\bibliographystyle{ACM-Reference-Format}
\bibliography{budgetedMechanismDesign.bib}


\begin{thebibliography}{44}


\ifx \showCODEN    \undefined \def \showCODEN     #1{\unskip}     \fi
\ifx \showDOI      \undefined \def \showDOI       #1{#1}\fi
\ifx \showISBNx    \undefined \def \showISBNx     #1{\unskip}     \fi
\ifx \showISBNxiii \undefined \def \showISBNxiii  #1{\unskip}     \fi
\ifx \showISSN     \undefined \def \showISSN      #1{\unskip}     \fi
\ifx \showLCCN     \undefined \def \showLCCN      #1{\unskip}     \fi
\ifx \shownote     \undefined \def \shownote      #1{#1}          \fi
\ifx \showarticletitle \undefined \def \showarticletitle #1{#1}   \fi
\ifx \showURL      \undefined \def \showURL       {\relax}        \fi
\providecommand\bibfield[2]{#2}
\providecommand\bibinfo[2]{#2}
\providecommand\natexlab[1]{#1}
\providecommand\showeprint[2][]{arXiv:#2}

\bibitem[\protect\citeauthoryear{Amanatidis, Birmpas, and Markakis}{Amanatidis
  et~al\mbox{.}}{2016}]%
        {ABM16}
\bibfield{author}{\bibinfo{person}{Georgios Amanatidis},
  \bibinfo{person}{Georgios Birmpas}, {and} \bibinfo{person}{Evangelos
  Markakis}.} \bibinfo{year}{2016}\natexlab{}.
\newblock \showarticletitle{Coverage, Matching, and Beyond: New Results on
  Budgeted Mechanism Design}. In \bibinfo{booktitle}{\emph{Web and Internet
  Economics - 12th International Conference, {WINE} 2016, Montreal, Canada,
  December 11-14, 2016, Proceedings}} \emph{(\bibinfo{series}{Lecture Notes in
  Computer Science})}, Vol.~\bibinfo{volume}{10123}.
  \bibinfo{publisher}{Springer}, \bibinfo{pages}{414--428}.
\newblock


\bibitem[\protect\citeauthoryear{Amanatidis, Birmpas, and Markakis}{Amanatidis
  et~al\mbox{.}}{2017}]%
        {ABM17}
\bibfield{author}{\bibinfo{person}{Georgios Amanatidis},
  \bibinfo{person}{Georgios Birmpas}, {and} \bibinfo{person}{Evangelos
  Markakis}.} \bibinfo{year}{2017}\natexlab{}.
\newblock \showarticletitle{On Budget-Feasible Mechanism Design for Symmetric
  Submodular Objectives}. In \bibinfo{booktitle}{\emph{Web and Internet
  Economics - 13th International Conference, {WINE} 2017, Bangalore, India,
  December 17-20, 2017, Proceedings}} \emph{(\bibinfo{series}{Lecture Notes in
  Computer Science})}, Vol.~\bibinfo{volume}{10660}.
  \bibinfo{publisher}{Springer}, \bibinfo{pages}{1--15}.
\newblock


\bibitem[\protect\citeauthoryear{Anari, Goel, and Nikzad}{Anari
  et~al\mbox{.}}{2014}]%
        {AnariGN14}
\bibfield{author}{\bibinfo{person}{Nima Anari}, \bibinfo{person}{Gagan Goel},
  {and} \bibinfo{person}{Afshin Nikzad}.} \bibinfo{year}{2014}\natexlab{}.
\newblock \showarticletitle{Mechanism Design for Crowdsourcing: An Optimal
  1-1/e Competitive Budget-Feasible Mechanism for Large Markets}. In
  \bibinfo{booktitle}{\emph{55th {IEEE} Annual Symposium on Foundations of
  Computer Science, {FOCS} 2014, Philadelphia, PA, USA, October 18-21, 2014}}.
  \bibinfo{pages}{266--275}.
\newblock


\bibitem[\protect\citeauthoryear{Babaioff, Immorlica, Kempe, and
  Kleinberg}{Babaioff et~al\mbox{.}}{2007}]%
        {BabaioffIKK07}
\bibfield{author}{\bibinfo{person}{Moshe Babaioff}, \bibinfo{person}{Nicole
  Immorlica}, \bibinfo{person}{David Kempe}, {and} \bibinfo{person}{Robert
  Kleinberg}.} \bibinfo{year}{2007}\natexlab{}.
\newblock \showarticletitle{A Knapsack Secretary Problem with Applications}. In
  \bibinfo{booktitle}{\emph{{APPROX-RANDOM}}} \emph{(\bibinfo{series}{Lecture
  Notes in Computer Science})}, Vol.~\bibinfo{volume}{4627}.
  \bibinfo{publisher}{Springer}, \bibinfo{pages}{16--28}.
\newblock


\bibitem[\protect\citeauthoryear{Badanidiyuru, Kleinberg, and
  Singer}{Badanidiyuru et~al\mbox{.}}{2012}]%
        {BadanidiyuruKS12}
\bibfield{author}{\bibinfo{person}{Ashwinkumar Badanidiyuru},
  \bibinfo{person}{Robert Kleinberg}, {and} \bibinfo{person}{Yaron Singer}.}
  \bibinfo{year}{2012}\natexlab{}.
\newblock \showarticletitle{Learning on a budget: posted price mechanisms for
  online procurement}. In \bibinfo{booktitle}{\emph{{EC}}}.
  \bibinfo{publisher}{{ACM}}, \bibinfo{pages}{128--145}.
\newblock


\bibitem[\protect\citeauthoryear{Badanidiyuru and Vondr{\'{a}}k}{Badanidiyuru
  and Vondr{\'{a}}k}{2014}]%
        {BadanidiyuruV14}
\bibfield{author}{\bibinfo{person}{Ashwinkumar Badanidiyuru} {and}
  \bibinfo{person}{Jan Vondr{\'{a}}k}.} \bibinfo{year}{2014}\natexlab{}.
\newblock \showarticletitle{Fast algorithms for maximizing submodular
  functions}. In \bibinfo{booktitle}{\emph{Proceedings of the Twenty-Fifth
  Annual {ACM-SIAM} Symposium on Discrete Algorithms, {SODA} 2014, Portland,
  Oregon, USA, January 5-7, 2014}}. \bibinfo{publisher}{{SIAM}},
  \bibinfo{pages}{1497--1514}.
\newblock


\bibitem[\protect\citeauthoryear{Balkanski and Hartline}{Balkanski and
  Hartline}{2016}]%
        {BalkanskiH16}
\bibfield{author}{\bibinfo{person}{Eric Balkanski} {and}
  \bibinfo{person}{Jason~D. Hartline}.} \bibinfo{year}{2016}\natexlab{}.
\newblock \showarticletitle{Bayesian Budget Feasibility with Posted Pricing}.
  In \bibinfo{booktitle}{\emph{Proceedings of the 25th International Conference
  on World Wide Web, {WWW} 2016, Montreal, Canada, April 11 - 15, 2016}}.
  \bibinfo{publisher}{{ACM}}, \bibinfo{pages}{189--203}.
\newblock


\bibitem[\protect\citeauthoryear{Bateni, Hajiaghayi, and Zadimoghaddam}{Bateni
  et~al\mbox{.}}{2013}]%
        {BateniHZ13}
\bibfield{author}{\bibinfo{person}{MohammadHossein Bateni},
  \bibinfo{person}{Mohammad~Taghi Hajiaghayi}, {and} \bibinfo{person}{Morteza
  Zadimoghaddam}.} \bibinfo{year}{2013}\natexlab{}.
\newblock \showarticletitle{Submodular secretary problem and extensions}.
\newblock \bibinfo{journal}{\emph{{ACM} Trans. Algorithms}}
  \bibinfo{volume}{9}, \bibinfo{number}{4} (\bibinfo{year}{2013}),
  \bibinfo{pages}{32:1--32:23}.
\newblock


\bibitem[\protect\citeauthoryear{Bei, Chen, Gravin, and Lu}{Bei
  et~al\mbox{.}}{2012}]%
        {BeiCGL12}
\bibfield{author}{\bibinfo{person}{Xiaohui Bei}, \bibinfo{person}{Ning Chen},
  \bibinfo{person}{Nick Gravin}, {and} \bibinfo{person}{Pinyan Lu}.}
  \bibinfo{year}{2012}\natexlab{}.
\newblock \showarticletitle{Budget feasible mechanism design: from prior-free
  to bayesian}. In \bibinfo{booktitle}{\emph{Proceedings of the 44th Symposium
  on Theory of Computing Conference, {STOC} 2012, New York, NY, USA, May 19 -
  22, 2012}}. \bibinfo{pages}{449--458}.
\newblock


\bibitem[\protect\citeauthoryear{Bei, Chen, Gravin, and Lu}{Bei
  et~al\mbox{.}}{2017}]%
        {BeiCGL17}
\bibfield{author}{\bibinfo{person}{Xiaohui Bei}, \bibinfo{person}{Ning Chen},
  \bibinfo{person}{Nick Gravin}, {and} \bibinfo{person}{Pinyan Lu}.}
  \bibinfo{year}{2017}\natexlab{}.
\newblock \showarticletitle{Worst-Case Mechanism Design via Bayesian Analysis}.
\newblock \bibinfo{journal}{\emph{{SIAM} J. Comput.}} \bibinfo{volume}{46},
  \bibinfo{number}{4} (\bibinfo{year}{2017}), \bibinfo{pages}{1428--1448}.
\newblock


\bibitem[\protect\citeauthoryear{Blumrosen and Nisan}{Blumrosen and
  Nisan}{2009}]%
        {BlumrosenN09}
\bibfield{author}{\bibinfo{person}{Liad Blumrosen} {and} \bibinfo{person}{Noam
  Nisan}.} \bibinfo{year}{2009}\natexlab{}.
\newblock \showarticletitle{On the Computational Power of Demand Queries}.
\newblock \bibinfo{journal}{\emph{{SIAM} J. Comput.}} \bibinfo{volume}{39},
  \bibinfo{number}{4} (\bibinfo{year}{2009}), \bibinfo{pages}{1372--1391}.
\newblock


\bibitem[\protect\citeauthoryear{Borodin, Filmus, and Oren}{Borodin
  et~al\mbox{.}}{2010}]%
        {BFO10}
\bibfield{author}{\bibinfo{person}{Allan Borodin}, \bibinfo{person}{Yuval
  Filmus}, {and} \bibinfo{person}{Joel Oren}.} \bibinfo{year}{2010}\natexlab{}.
\newblock \showarticletitle{Threshold Models for Competitive Influence in
  Social Networks}. In \bibinfo{booktitle}{\emph{Proceedings of the 6th
  International Workshop on Internet and Network Economics, {WINE 2010}}}.
  \bibinfo{pages}{539--550}.
\newblock


\bibitem[\protect\citeauthoryear{Buchbinder and Feldman}{Buchbinder and
  Feldman}{2018}]%
        {BuchbinderF18}
\bibfield{author}{\bibinfo{person}{Niv Buchbinder} {and} \bibinfo{person}{Moran
  Feldman}.} \bibinfo{year}{2018}\natexlab{}.
\newblock \showarticletitle{Deterministic Algorithms for Submodular
  Maximization Problems}.
\newblock \bibinfo{journal}{\emph{{ACM} Trans. Algorithms}}
  \bibinfo{volume}{14}, \bibinfo{number}{3} (\bibinfo{year}{2018}),
  \bibinfo{pages}{32:1--32:20}.
\newblock


\bibitem[\protect\citeauthoryear{Buchbinder, Feldman, Naor, and
  Schwartz}{Buchbinder et~al\mbox{.}}{2014}]%
        {BuchbinderFNS14}
\bibfield{author}{\bibinfo{person}{Niv Buchbinder}, \bibinfo{person}{Moran
  Feldman}, \bibinfo{person}{Joseph Naor}, {and} \bibinfo{person}{Roy
  Schwartz}.} \bibinfo{year}{2014}\natexlab{}.
\newblock \showarticletitle{Submodular Maximization with Cardinality
  Constraints}. In \bibinfo{booktitle}{\emph{Proceedings of the Twenty-Fifth
  Annual {ACM-SIAM} Symposium on Discrete Algorithms, {SODA} 2014, Portland,
  Oregon, USA, January 5-7, 2014}}. \bibinfo{publisher}{{SIAM}},
  \bibinfo{pages}{1433--1452}.
\newblock


\bibitem[\protect\citeauthoryear{Chekuri, Gupta, and Quanrud}{Chekuri
  et~al\mbox{.}}{2015}]%
        {CGQ15}
\bibfield{author}{\bibinfo{person}{Chandra Chekuri}, \bibinfo{person}{Shalmoli
  Gupta}, {and} \bibinfo{person}{Kent Quanrud}.}
  \bibinfo{year}{2015}\natexlab{}.
\newblock \showarticletitle{Streaming Algorithms for Submodular Function
  Maximization}. In \bibinfo{booktitle}{\emph{Automata, Languages, and
  Programming - 42nd International Colloquium, {ICALP} 2015, Kyoto, Japan, July
  6-10, 2015, Proceedings, Part {I}}} \emph{(\bibinfo{series}{Lecture Notes in
  Computer Science})}, Vol.~\bibinfo{volume}{9134}.
  \bibinfo{publisher}{Springer}, \bibinfo{pages}{318--330}.
\newblock


\bibitem[\protect\citeauthoryear{Chekuri, Vondr{\'{a}}k, and Zenklusen}{Chekuri
  et~al\mbox{.}}{2014}]%
        {ChekuriVZ14}
\bibfield{author}{\bibinfo{person}{Chandra Chekuri}, \bibinfo{person}{Jan
  Vondr{\'{a}}k}, {and} \bibinfo{person}{Rico Zenklusen}.}
  \bibinfo{year}{2014}\natexlab{}.
\newblock \showarticletitle{Submodular Function Maximization via the
  Multilinear Relaxation and Contention Resolution Schemes}.
\newblock \bibinfo{journal}{\emph{{SIAM} J. Comput.}} \bibinfo{volume}{43},
  \bibinfo{number}{6} (\bibinfo{year}{2014}), \bibinfo{pages}{1831--1879}.
\newblock


\bibitem[\protect\citeauthoryear{Chen, Gravin, and Lu}{Chen
  et~al\mbox{.}}{2011}]%
        {ChenGL11}
\bibfield{author}{\bibinfo{person}{Ning Chen}, \bibinfo{person}{Nick Gravin},
  {and} \bibinfo{person}{Pinyan Lu}.} \bibinfo{year}{2011}\natexlab{}.
\newblock \showarticletitle{On the Approximability of Budget Feasible
  Mechanisms}. In \bibinfo{booktitle}{\emph{Proceedings of the Twenty-Second
  Annual {ACM-SIAM} Symposium on Discrete Algorithms, {SODA} 2011, San
  Francisco, California, USA, January 23-25, 2011}}. \bibinfo{pages}{685--699}.
\newblock


\bibitem[\protect\citeauthoryear{Dobzinski, Papadimitriou, and
  Singer}{Dobzinski et~al\mbox{.}}{2011}]%
        {DobzinskiPS11}
\bibfield{author}{\bibinfo{person}{Shahar Dobzinski},
  \bibinfo{person}{Christos~H. Papadimitriou}, {and} \bibinfo{person}{Yaron
  Singer}.} \bibinfo{year}{2011}\natexlab{}.
\newblock \showarticletitle{Mechanisms for complement-free procurement}. In
  \bibinfo{booktitle}{\emph{Proceedings 12th {ACM} Conference on Electronic
  Commerce (EC-2011), San Jose, CA, USA, June 5-9, 2011}}.
  \bibinfo{pages}{273--282}.
\newblock


\bibitem[\protect\citeauthoryear{Dynkin}{Dynkin}{1963}]%
        {Dynkin1963}
\bibfield{author}{\bibinfo{person}{Evgenii~Borisovich Dynkin}.}
  \bibinfo{year}{1963}\natexlab{}.
\newblock \showarticletitle{Optimal choice of the stopping moment of a Markov
  process}. In \bibinfo{booktitle}{\emph{Doklady Akademii Nauk}},
  Vol.~\bibinfo{volume}{150}. Russian Academy of Sciences,
  \bibinfo{pages}{238--240}.
\newblock


\bibitem[\protect\citeauthoryear{Ene and Nguyen}{Ene and Nguyen}{2017}]%
        {EneN17}
\bibfield{author}{\bibinfo{person}{Alina Ene} {and} \bibinfo{person}{Huy~L.
  Nguyen}.} \bibinfo{year}{2017}\natexlab{}.
\newblock \showarticletitle{A Nearly-linear Time Algorithm for Submodular
  Maximization with a Knapsack Constraint}.
\newblock \bibinfo{journal}{\emph{CoRR}}  \bibinfo{volume}{abs/1709.09767}
  (\bibinfo{year}{2017}).
\newblock
\showeprint[arxiv]{1709.09767}
\urldef\tempurl%
\url{http://arxiv.org/abs/1709.09767}
\showURL{%
\tempurl}


\bibitem[\protect\citeauthoryear{Feige, Mirrokni, and Vondr{\'{a}}k}{Feige
  et~al\mbox{.}}{2011}]%
        {FeigeMV11}
\bibfield{author}{\bibinfo{person}{Uriel Feige}, \bibinfo{person}{Vahab~S.
  Mirrokni}, {and} \bibinfo{person}{Jan Vondr{\'{a}}k}.}
  \bibinfo{year}{2011}\natexlab{}.
\newblock \showarticletitle{Maximizing Non-monotone Submodular Functions}.
\newblock \bibinfo{journal}{\emph{{SIAM} J. Comput.}} \bibinfo{volume}{40},
  \bibinfo{number}{4} (\bibinfo{year}{2011}), \bibinfo{pages}{1133--1153}.
\newblock


\bibitem[\protect\citeauthoryear{Feldman, Harshaw, and Karbasi}{Feldman
  et~al\mbox{.}}{2017}]%
        {FeldmanHK17}
\bibfield{author}{\bibinfo{person}{Moran Feldman}, \bibinfo{person}{Christopher
  Harshaw}, {and} \bibinfo{person}{Amin Karbasi}.}
  \bibinfo{year}{2017}\natexlab{}.
\newblock \showarticletitle{Greed Is Good: Near-Optimal Submodular Maximization
  via Greedy Optimization}. In \bibinfo{booktitle}{\emph{Proceedings of the
  30th Conference on Learning Theory, {COLT} 2017, Amsterdam, The Netherlands,
  7-10 July 2017}} \emph{(\bibinfo{series}{Proceedings of Machine Learning
  Research})}, Vol.~\bibinfo{volume}{65}. \bibinfo{publisher}{{PMLR}},
  \bibinfo{pages}{758--784}.
\newblock


\bibitem[\protect\citeauthoryear{Feldman, Naor, and Schwartz}{Feldman
  et~al\mbox{.}}{2011a}]%
        {FeldmanNS11b}
\bibfield{author}{\bibinfo{person}{Moran Feldman}, \bibinfo{person}{Joseph
  Naor}, {and} \bibinfo{person}{Roy Schwartz}.}
  \bibinfo{year}{2011}\natexlab{a}.
\newblock \showarticletitle{Improved Competitive Ratios for Submodular
  Secretary Problems (Extended Abstract)}. In
  \bibinfo{booktitle}{\emph{{APPROX-RANDOM}}} \emph{(\bibinfo{series}{Lecture
  Notes in Computer Science})}, Vol.~\bibinfo{volume}{6845}.
  \bibinfo{publisher}{Springer}, \bibinfo{pages}{218--229}.
\newblock


\bibitem[\protect\citeauthoryear{Feldman, Naor, and Schwartz}{Feldman
  et~al\mbox{.}}{2011b}]%
        {FeldmanNS11}
\bibfield{author}{\bibinfo{person}{Moran Feldman}, \bibinfo{person}{Joseph
  Naor}, {and} \bibinfo{person}{Roy Schwartz}.}
  \bibinfo{year}{2011}\natexlab{b}.
\newblock \showarticletitle{A Unified Continuous Greedy Algorithm for
  Submodular Maximization}. In \bibinfo{booktitle}{\emph{{IEEE} 52nd Annual
  Symposium on Foundations of Computer Science, {FOCS} 2011, Palm Springs, CA,
  USA, October 22-25, 2011}}. \bibinfo{publisher}{{IEEE} Computer Society},
  \bibinfo{pages}{570--579}.
\newblock


\bibitem[\protect\citeauthoryear{Feldman and Zenklusen}{Feldman and
  Zenklusen}{2018}]%
        {FeldmanZ18}
\bibfield{author}{\bibinfo{person}{Moran Feldman} {and} \bibinfo{person}{Rico
  Zenklusen}.} \bibinfo{year}{2018}\natexlab{}.
\newblock \showarticletitle{The Submodular Secretary Problem Goes Linear}.
\newblock \bibinfo{journal}{\emph{{SIAM} J. Comput.}} \bibinfo{volume}{47},
  \bibinfo{number}{2} (\bibinfo{year}{2018}), \bibinfo{pages}{330--366}.
\newblock


\bibitem[\protect\citeauthoryear{Goel, Nikzad, and Singla}{Goel
  et~al\mbox{.}}{2014}]%
        {GoelNS14}
\bibfield{author}{\bibinfo{person}{Gagan Goel}, \bibinfo{person}{Afshin
  Nikzad}, {and} \bibinfo{person}{Adish Singla}.}
  \bibinfo{year}{2014}\natexlab{}.
\newblock \showarticletitle{Mechanism Design for Crowdsourcing Markets with
  Heterogeneous Tasks}. In \bibinfo{booktitle}{\emph{Proceedings of the Seconf
  {AAAI} Conference on Human Computation and Crowdsourcing, {HCOMP} 2014,
  November 2-4, 2014, Pittsburgh, Pennsylvania, {USA}}}.
\newblock


\bibitem[\protect\citeauthoryear{Gupta, Nagarajan, and Singla}{Gupta
  et~al\mbox{.}}{2017}]%
        {GuptaNS17}
\bibfield{author}{\bibinfo{person}{Anupam Gupta}, \bibinfo{person}{Viswanath
  Nagarajan}, {and} \bibinfo{person}{Sahil Singla}.}
  \bibinfo{year}{2017}\natexlab{}.
\newblock \showarticletitle{Adaptivity Gaps for Stochastic Probing: Submodular
  and {XOS} Functions}. In \bibinfo{booktitle}{\emph{Proceedings of the
  Twenty-Eighth Annual {ACM-SIAM} Symposium on Discrete Algorithms, {SODA}
  2017, Barcelona, Spain, Hotel Porta Fira, January 16-19}}.
  \bibinfo{publisher}{{SIAM}}, \bibinfo{pages}{1688--1702}.
\newblock


\bibitem[\protect\citeauthoryear{Gupta, Roth, Schoenebeck, and Talwar}{Gupta
  et~al\mbox{.}}{2010}]%
        {GuptaRST10}
\bibfield{author}{\bibinfo{person}{Anupam Gupta}, \bibinfo{person}{Aaron Roth},
  \bibinfo{person}{Grant Schoenebeck}, {and} \bibinfo{person}{Kunal Talwar}.}
  \bibinfo{year}{2010}\natexlab{}.
\newblock \showarticletitle{Constrained Non-monotone Submodular Maximization:
  Offline and Secretary Algorithms}. In \bibinfo{booktitle}{\emph{Internet and
  Network Economics - 6th International Workshop, {WINE} 2010, Stanford, CA,
  USA, December 13-17, 2010. Proceedings}} \emph{(\bibinfo{series}{LNCS})},
  Vol.~\bibinfo{volume}{6484}. \bibinfo{publisher}{Springer},
  \bibinfo{pages}{246--257}.
\newblock


\bibitem[\protect\citeauthoryear{Horel, Ioannidis, and Muthukrishnan}{Horel
  et~al\mbox{.}}{2014}]%
        {HorelIM14}
\bibfield{author}{\bibinfo{person}{Thibaut Horel}, \bibinfo{person}{Stratis
  Ioannidis}, {and} \bibinfo{person}{S. Muthukrishnan}.}
  \bibinfo{year}{2014}\natexlab{}.
\newblock \showarticletitle{Budget Feasible Mechanisms for Experimental
  Design}. In \bibinfo{booktitle}{\emph{{LATIN} 2014: Theoretical Informatics -
  11th Latin American Symposium, Montevideo, Uruguay, March 31 - April 4, 2014.
  Proceedings}}. \bibinfo{pages}{719--730}.
\newblock


\bibitem[\protect\citeauthoryear{Kesselheim and T{\"{o}}nnis}{Kesselheim and
  T{\"{o}}nnis}{2017}]%
        {KesselheimT17}
\bibfield{author}{\bibinfo{person}{Thomas Kesselheim} {and}
  \bibinfo{person}{Andreas T{\"{o}}nnis}.} \bibinfo{year}{2017}\natexlab{}.
\newblock \showarticletitle{Submodular Secretary Problems: Cardinality,
  Matching, and Linear Constraints}. In
  \bibinfo{booktitle}{\emph{{APPROX-RANDOM}}}
  \emph{(\bibinfo{series}{LIPIcs})}, Vol.~\bibinfo{volume}{81}.
  \bibinfo{publisher}{Schloss Dagstuhl - Leibniz-Zentrum fuer Informatik},
  \bibinfo{pages}{16:1--16:22}.
\newblock


\bibitem[\protect\citeauthoryear{Khalilabadi and Tardos}{Khalilabadi and
  Tardos}{2018}]%
        {JT17}
\bibfield{author}{\bibinfo{person}{Pooya~Jalaly Khalilabadi} {and}
  \bibinfo{person}{{\'{E}}va Tardos}.} \bibinfo{year}{2018}\natexlab{}.
\newblock \showarticletitle{Simple and Efficient Budget Feasible Mechanisms for
  Monotone Submodular Valuations}. In \bibinfo{booktitle}{\emph{Web and
  Internet Economics - 14th International Conference, {WINE} 2018, Oxford, UK,
  December 15-17, 2018, Proceedings}} \emph{(\bibinfo{series}{Lecture Notes in
  Computer Science})}, Vol.~\bibinfo{volume}{11316}.
  \bibinfo{publisher}{Springer}, \bibinfo{pages}{246--263}.
\newblock


\bibitem[\protect\citeauthoryear{Kulik, Shachnai, and Tamir}{Kulik
  et~al\mbox{.}}{2013}]%
        {KulikST13}
\bibfield{author}{\bibinfo{person}{Ariel Kulik}, \bibinfo{person}{Hadas
  Shachnai}, {and} \bibinfo{person}{Tami Tamir}.}
  \bibinfo{year}{2013}\natexlab{}.
\newblock \showarticletitle{Approximations for Monotone and Nonmonotone
  Submodular Maximization with Knapsack Constraints}.
\newblock \bibinfo{journal}{\emph{Math. Oper. Res.}} \bibinfo{volume}{38},
  \bibinfo{number}{4} (\bibinfo{year}{2013}), \bibinfo{pages}{729--739}.
\newblock


\bibitem[\protect\citeauthoryear{Lehmann, Lehmann, and Nisan}{Lehmann
  et~al\mbox{.}}{2006}]%
        {LehmannLN06}
\bibfield{author}{\bibinfo{person}{Benny Lehmann}, \bibinfo{person}{Daniel~J.
  Lehmann}, {and} \bibinfo{person}{Noam Nisan}.}
  \bibinfo{year}{2006}\natexlab{}.
\newblock \showarticletitle{Combinatorial auctions with decreasing marginal
  utilities}.
\newblock \bibinfo{journal}{\emph{Games and Economic Behavior}}
  \bibinfo{volume}{55}, \bibinfo{number}{2} (\bibinfo{year}{2006}),
  \bibinfo{pages}{270--296}.
\newblock


\bibitem[\protect\citeauthoryear{Leonardi, Monaco, Sankowski, and
  Zhang}{Leonardi et~al\mbox{.}}{2016}]%
        {LeonardiMSZ16}
\bibfield{author}{\bibinfo{person}{Stefano Leonardi},
  \bibinfo{person}{Gianpiero Monaco}, \bibinfo{person}{Piotr Sankowski}, {and}
  \bibinfo{person}{Qiang Zhang}.} \bibinfo{year}{2016}\natexlab{}.
\newblock \showarticletitle{Budget Feasible Mechanisms on Matroids}.
\newblock \bibinfo{journal}{\emph{CoRR}}  \bibinfo{volume}{abs/1612.03150}
  (\bibinfo{year}{2016}).
\newblock


\bibitem[\protect\citeauthoryear{Leonardi, Monaco, Sankowski, and
  Zhang}{Leonardi et~al\mbox{.}}{2017}]%
        {LeonardiMSZ17}
\bibfield{author}{\bibinfo{person}{Stefano Leonardi},
  \bibinfo{person}{Gianpiero Monaco}, \bibinfo{person}{Piotr Sankowski}, {and}
  \bibinfo{person}{Qiang Zhang}.} \bibinfo{year}{2017}\natexlab{}.
\newblock \showarticletitle{Budget Feasible Mechanisms on Matroids}. In
  \bibinfo{booktitle}{\emph{Integer Programming and Combinatorial Optimization
  - 19th International Conference, {IPCO} 2017, Waterloo, ON, Canada, June
  26-28, 2017, Proceedings}} \emph{(\bibinfo{series}{Lecture Notes in Computer
  Science})}, Vol.~\bibinfo{volume}{10328}. \bibinfo{publisher}{Springer},
  \bibinfo{pages}{368--379}.
\newblock


\bibitem[\protect\citeauthoryear{Mirrokni, Schapira, and
  Vondr{\'{a}}k}{Mirrokni et~al\mbox{.}}{2008}]%
        {MirrokniSV08}
\bibfield{author}{\bibinfo{person}{Vahab~S. Mirrokni}, \bibinfo{person}{Michael
  Schapira}, {and} \bibinfo{person}{Jan Vondr{\'{a}}k}.}
  \bibinfo{year}{2008}\natexlab{}.
\newblock \showarticletitle{Tight information-theoretic lower bounds for
  welfare maximization in combinatorial auctions}. In
  \bibinfo{booktitle}{\emph{Proceedings 9th {ACM} Conference on Electronic
  Commerce (EC-2008), Chicago, IL, USA, June 8-12, 2008}}.
  \bibinfo{publisher}{{ACM}}, \bibinfo{pages}{70--77}.
\newblock


\bibitem[\protect\citeauthoryear{Mirzasoleiman, Badanidiyuru, and
  Karbasi}{Mirzasoleiman et~al\mbox{.}}{2016}]%
        {MirzasoleimanBK16}
\bibfield{author}{\bibinfo{person}{Baharan Mirzasoleiman},
  \bibinfo{person}{Ashwinkumar Badanidiyuru}, {and} \bibinfo{person}{Amin
  Karbasi}.} \bibinfo{year}{2016}\natexlab{}.
\newblock \showarticletitle{Fast Constrained Submodular Maximization:
  Personalized Data Summarization}. In \bibinfo{booktitle}{\emph{Proceedings of
  the 33nd International Conference on Machine Learning, {ICML} 2016, New York
  City, NY, USA, June 19-24, 2016}} \emph{(\bibinfo{series}{JMLR})},
  Vol.~\bibinfo{volume}{48}. \bibinfo{publisher}{JMLR.org},
  \bibinfo{pages}{1358--1367}.
\newblock


\bibitem[\protect\citeauthoryear{Myerson}{Myerson}{1981}]%
        {Myerson81}
\bibfield{author}{\bibinfo{person}{Roger Myerson}.}
  \bibinfo{year}{1981}\natexlab{}.
\newblock \showarticletitle{Optimal Auction Design}.
\newblock \bibinfo{journal}{\emph{Mathematics of Operations Research}}
  \bibinfo{volume}{6}, \bibinfo{number}{1} (\bibinfo{year}{1981}).
\newblock


\bibitem[\protect\citeauthoryear{Nemhauser, Wolsey, and Fisher}{Nemhauser
  et~al\mbox{.}}{1978}]%
        {NemhauserWF78}
\bibfield{author}{\bibinfo{person}{George~L. Nemhauser},
  \bibinfo{person}{Laurence~A. Wolsey}, {and} \bibinfo{person}{Marshall~L.
  Fisher}.} \bibinfo{year}{1978}\natexlab{}.
\newblock \showarticletitle{An analysis of approximations for maximizing
  submodular set functions - {I}}.
\newblock \bibinfo{journal}{\emph{Math. Program.}} \bibinfo{volume}{14},
  \bibinfo{number}{1} (\bibinfo{year}{1978}), \bibinfo{pages}{265--294}.
\newblock


\bibitem[\protect\citeauthoryear{Singer}{Singer}{2010}]%
        {Singer10}
\bibfield{author}{\bibinfo{person}{Yaron Singer}.}
  \bibinfo{year}{2010}\natexlab{}.
\newblock \showarticletitle{Budget Feasible Mechanisms}. In
  \bibinfo{booktitle}{\emph{51th Annual {IEEE} Symposium on Foundations of
  Computer Science, {FOCS} 2010, October 23-26, 2010, Las Vegas, Nevada,
  {USA}}}. \bibinfo{publisher}{{IEEE} Computer Society},
  \bibinfo{pages}{765--774}.
\newblock


\bibitem[\protect\citeauthoryear{Singer}{Singer}{2012}]%
        {Singer12}
\bibfield{author}{\bibinfo{person}{Yaron Singer}.}
  \bibinfo{year}{2012}\natexlab{}.
\newblock \showarticletitle{How to win friends and influence people,
  truthfully: influence maximization mechanisms for social networks}. In
  \bibinfo{booktitle}{\emph{Proceedings of the Fifth International Conference
  on Web Search and Web Data Mining, {WSDM} 2012, Seattle, WA, USA, February
  8-12, 2012}}. \bibinfo{pages}{733--742}.
\newblock


\bibitem[\protect\citeauthoryear{Singla and Krause}{Singla and Krause}{2013}]%
        {SinglaK13}
\bibfield{author}{\bibinfo{person}{Adish Singla} {and} \bibinfo{person}{Andreas
  Krause}.} \bibinfo{year}{2013}\natexlab{}.
\newblock \showarticletitle{Incentives for Privacy Tradeoff in Community
  Sensing}. In \bibinfo{booktitle}{\emph{Proceedings of the First {AAAI}
  Conference on Human Computation and Crowdsourcing, {HCOMP} 2013, November
  7-9, 2013, Palm Springs, CA, {USA}}}. \bibinfo{publisher}{{AAAI}}.
\newblock


\bibitem[\protect\citeauthoryear{Sviridenko}{Sviridenko}{2004}]%
        {Sviridenko04}
\bibfield{author}{\bibinfo{person}{Maxim Sviridenko}.}
  \bibinfo{year}{2004}\natexlab{}.
\newblock \showarticletitle{A note on maximizing a submodular set function
  subject to a knapsack constraint}.
\newblock \bibinfo{journal}{\emph{Oper. Res. Lett.}} \bibinfo{volume}{32},
  \bibinfo{number}{1} (\bibinfo{year}{2004}), \bibinfo{pages}{41--43}.
\newblock


\bibitem[\protect\citeauthoryear{Wolsey}{Wolsey}{1982}]%
        {Wolsey82}
\bibfield{author}{\bibinfo{person}{Laurence~A. Wolsey}.}
  \bibinfo{year}{1982}\natexlab{}.
\newblock \showarticletitle{Maximising Real-Valued Submodular Functions: Primal
  and Dual Heuristics for Location Problems}.
\newblock \bibinfo{journal}{\emph{Math. Oper. Res.}} \bibinfo{volume}{7},
  \bibinfo{number}{3} (\bibinfo{year}{1982}), \bibinfo{pages}{410--425}.
\newblock


\end{thebibliography}

\end{document}